\documentclass{article}

\usepackage{amsmath,amsfonts,bm}

\newcommand{\newterm}[1]{{\bf #1}}

\def\1{\bm{1}}

\def\va{{\bm{a}}}

\def\vu{{\bm{u}}}

\def\vw{{\bm{w}}}
\def\vx{{\bm{x}}}

\def\mA{{\bm{A}}}

\def\mF{{\bm{F}}}

\def\mI{{\bm{I}}}

\def\mP{{\bm{P}}}

\def\mR{{\bm{R}}}

\def\mU{{\bm{U}}}

\def\mW{{\bm{W}}}

\DeclareMathAlphabet{\mathsfit}{\encodingdefault}{\sfdefault}{m}{sl}
\SetMathAlphabet{\mathsfit}{bold}{\encodingdefault}{\sfdefault}{bx}{n}

\def\sC{{\mathbb{C}}}

\def\sR{{\mathbb{R}}}

\def\sW{{\mathbb{W}}}

\def\sZ{{\mathbb{Z}}}

\DeclareMathOperator*{\argmin}{arg\,min}

\DeclareMathOperator{\Tr}{Tr}

\usepackage{microtype}
\usepackage{graphicx}
\usepackage{subfigure}
\usepackage{booktabs}

\usepackage{hyperref}
\usepackage{xurl}

\usepackage{nicefrac}
\usepackage{braket}
\usepackage{algorithm, algorithmic}
\usepackage[whole]{bxcjkjatype}
\usepackage{mathtools}
\usepackage{amsmath, amssymb, amsthm}
\usepackage[capitalize,noabbrev]{cleveref}
\usepackage{xcolor}
\allowdisplaybreaks[4]

\DeclareMathOperator{\poly}{poly}
\DeclareMathOperator{\polylog}{polylog}

\usepackage[accepted]{icml2023}

\theoremstyle{plain}
\newtheorem{theorem}{Theorem}[section]

\theoremstyle{definition}

\theoremstyle{remark}

\usepackage[textsize=tiny]{todonotes}

\icmltitlerunning{Quantum Ridgelet Transform: Winning Lottery Ticket of Neural Networks with Quantum Computation}

\begin{document}

\twocolumn[
\icmltitle{Quantum Ridgelet Transform:\\
           Winning Lottery Ticket of Neural Networks with Quantum Computation}



\icmlsetsymbol{equal}{*}

\begin{icmlauthorlist}
\icmlauthor{Hayata Yamasaki}{ut}
\icmlauthor{Sathyawageeswar Subramanian}{warwick}
\icmlauthor{Satoshi Hayakawa}{oxford}
\icmlauthor{Sho Sonoda}{riken}
\end{icmlauthorlist}

\icmlaffiliation{ut}{The University of Tokyo}
\icmlaffiliation{warwick}{University of Warwick}
\icmlaffiliation{oxford}{University of Oxford}
\icmlaffiliation{riken}{RIKEN AIP}

\icmlcorrespondingauthor{Hayata Yamasaki}{hayata.yamasaki@phys.s.u-tokyo.ac.jp}

\icmlkeywords{quantum machine learning, lottery ticket hypothesis, neural network, quantum algorithm, ridgelet transform}

\vskip 0.3in
]



\printAffiliationsAndNotice{}  

\begin{abstract}
    A significant challenge in the field of quantum machine learning (QML) is to establish applications of quantum computation to accelerate common tasks in machine learning such as those for neural networks. Ridgelet transform has been a fundamental mathematical tool in the theoretical studies of neural networks, but the practical applicability of ridgelet transform to conducting learning tasks was limited since its numerical implementation by conventional classical computation requires an exponential runtime $\exp(O(D))$ as data dimension $D$ increases. To address this problem, we develop a quantum ridgelet transform (QRT), which implements the ridgelet transform of a quantum state within a linear runtime $O(D)$ of quantum computation. As an application, we also show that one can use QRT as a fundamental subroutine for QML to efficiently find a sparse trainable subnetwork of large shallow wide neural networks without conducting large-scale optimization of the original network. This application discovers an efficient way in this regime to demonstrate the lottery ticket hypothesis on finding such a sparse trainable neural network. These results open an avenue of QML for accelerating learning tasks with commonly used classical neural networks.
\end{abstract}

\section{Introduction}

Quantum machine learning (QML) is an emerging field of research to take advantage of quantum computation for accelerating machine-learning tasks~\citep{biamonte2017quantum,doi:10.1098/rspa.2017.0551,schuld2021machine}.
Quantum computation can achieve significant speedups compared to the best existing algorithms with conventional classical computation in solving various computational tasks~\cite{N4,arXiv:1907.09415}, such as Shor's algorithm for integer factorization~\citep{doi:10.1137/S0097539795293172}.
QML indeed has advantages in learning data obtained from quantum states~\citep{Sweke2021quantumversus,PhysRevLett.126.190505,doi:10.1126/science.abn7293,9719827}, yet machine learning commonly deals with classical data rather than quantum states.
For a classical dataset constructed carefully so that its classification reduces to a variant of Shor's algorithm, QML achieves the classification superpolynomially faster than classical algorithms~\citep{Yunchao2020}; however, the applicability of such QML to practical datasets has been unknown.
Meanwhile, motivated by the success of neural networks~\citep{Goodfellow-et-al-2016}, various attempts have been made to apply quantum computation to more practical tasks for neural networks.
For example, one widely studied approach in QML is to use parameterized quantum circuits, often called ``quantum neural networks'', as a potential substitute for conventional classical neural networks; however, problematically, the parameterized quantum circuits do not successfully emulate essential components of the neural networks, e.g., perceptrons and nonlinear activation functions, due to linearity of the transformation implemented by the quantum circuits~\citep{schuld2021machine}.
Thus, a significant challenge in QML has been to develop a novel technique to bridge the gap between quantum computation and classical neural networks, so as to clarify what advantage QML could offer on top of the empirically proven merit of the classical neural networks.

To address this challenge, we here develop a fundamental quantum algorithm for making the tasks for classical neural networks more efficient, based on ridgelet transform.
Ridgelet transform, one of the well-studied integral transforms in signal processing, is a fundamental mathematical tool for studying neural networks in the over-parameterized regime~\citep{MURATA1996947,candes,Rubin1998,starck_murtagh_fadili_2010,pmlr-v130-sonoda21a,pmlr-v162-sonoda22a,sonoda2022universality}.
Let $f:\mathbb{R}^D\to\mathbb{R}$ denote a function with $D$-dimensional input, to
be learned with a neural network. For an activation function
$g:\mathbb{R}\to\mathbb{R}$ such as the rectified linear unit (ReLU), a shallow
feed-forward neural network with a single hidden layer is represented by
$f(\vx)\approx\sum_{n=1}^{N}w_n g(\va_n^\top \vx-b_n)$,
where $N$ is the number of nodes in the hidden layer, and $w_n$ is the weight of the map $g(\va_n^\top \vx-b_n)$ parameterized by $(\va_n,b_n)$ at node $n\in\{1,\ldots,N\}$~\citep{Goodfellow-et-al-2016}.
In the over-parameterized (continuous) limit $N\to\infty$, the representation simplifies into an \newterm{integral representation} of the neural network~\citep{256500,MURATA1996947,candes,SONODA2017233}, i.e.,
\begin{equation} \label{eq:continuous}
    \!f(\vx)=S[w](\vx)\coloneqq\int_{\sR^D\times\sR}\!\!d\va\,db\,w(\va,b)g(\va^{\top}
\vx-b),
\end{equation}
where $(\va,b)$ runs over all possible parameters in the continuous space, and $w:\mathbb{R}^D\times\mathbb{R}\to\mathbb{R}$ at each $(\va,b)$ corresponds to the weight $w_n$ at the node $n$ with parameter $(\va_{n},b_n)=(\va,b)$.
With a ridgelet function $r:\mathbb{R}^D\to\mathbb{R}$ that we appropriately choose corresponding to $g$, the $D$-dimensional \newterm{ridgelet transform} $R[f]$ is defined as an inverse transform of $S[w]$ in~\eqref{eq:continuous}, characterizing a weight $w$ to represent $f$, given by
\begin{equation}
    \label{eq:ridgelet}
w(\va,b)=R[f](\va,b)\coloneqq\int_{\sR^D}d\vx\, f(\vx)r(\va^\top \vx-b).
\end{equation}
A wide class of function $f$ is known to be representable as~\eqref{eq:continuous}; moreover, if $g$ and $r$ satisfy a certain admissibility condition, we can reconstruct $f$ from the ridgelet transform of $f$, i.e., $f\propto S[R[f]]$, up to a normalization factor~\citep{SONODA2017233}.
For theoretical analysis, an essential benefit of the integral representation is to simplify the analysis by the linearity; that is, we can regard~\eqref{eq:continuous} as the linear combination of an non-orthogonal over-complete basis of functions, i.e., $\{g(\va^\top \vx-b):(\va,b)\in\sR^D\times\sR\}$,
with weight $w(\va,b)$ given by the ridgelet transform of $f$.

Progressing beyond using the ridgelet transform for theoretical analysis, our key idea is to study its use for conducting tasks for neural networks.
However, $D$-dimensional ridgelet transform has been computationally hard to use in practice since the existing algorithms for ridgelet transform with conventional classical computation require $\exp(O(D))$ runtime as $D$ increases~\citep{1187351,CARRE20042165,4011958,10.1007/978-3-319-11179-7_68}.
After all, the $D$-dimensional ridgelet transform is a transform of $D$-dimensional functions in an $\exp(O(D))$-size space (see Sec.~\ref{sec:ridgelet} for detail), and classical algorithms for such transforms conventionally need $\exp(O(D))$ runtime; e.g., fast Fourier transform may be a more established transform algorithm but still needs $O(n\log(n))=\exp(O(D))$ runtime for the space of size $n=\exp(O(D))$.
To solve these problems, we discover that we can employ quantum computation.
Our results are as follows.
\begin{enumerate}
    \item (Sec.~\ref{sec:ridgelet}) To make exact implementation of ridgelet transform possible for computer with a finite number of bits and quantum bits (qubits), we formulate a new discretized version of ridgelet transform, which we call \newterm{discrete ridgelet transform}. We prove that our discretized ridgelet transform can be used for exactly representing any function on the discretized domain.
    \item (Sec.~\ref{sec:quantum_ridgelet_transform})  We develop a quantum algorithm to apply the $D$-dimensional discrete ridgelet transform to a quantum state of $O(D)$ qubits, i.e., a state in an $\exp(O(D))$-dimensional space, only within  linear runtime $O(D)$. We call this quantum algorithm \newterm{quantum ridgelet transform (QRT)}.
    QRT is exponentially faster in $D$ than the $\exp(O(D))$ runtime of the best existing classical algorithm for ridgelet transform in the $\exp(O(D))$-size space, in the same spirit as \newterm{quantum Fourier transform (QFT)}~\citep{https://doi.org/10.48550/arxiv.quant-ph/0201067} being exponentially faster than the corresponding classical algorithm of fast Fourier transform.
\item (Sec.~\ref{sec:application}) As an application, we demonstrate that we can use QRT to learn a sparse representation of an unknown function $f$ by sampling a subnetwork of a shallow wide neural network to approximate $f$ well. We analytically show the advantageous cases of our algorithm and also conduct a numerical simulation to support the advantage. This application is important as a demonstration of the \newterm{lottery ticket hypothesis}~\citep{frankle2018the}, as explained in the following.
\end{enumerate}

\paragraph{Contribution to QML with neural networks}

State-of-the-art neural networks have billions of parameters to attain high learning accuracy, but such large-scale networks may be problematic for practical use, e.g., with mobile devices and embedded systems.
Pruning techniques for neural networks gain growing importance in learning with neural networks.
The lottery ticket hypothesis by~\citet{frankle2018the} claims that, in such a large-scale neural network, one can find a sparse trainable subnetwork.
However, it is computationally demanding to search for the appropriate subnetwork in the large-scale neural network.

To apply QML to this pruning problem, our idea is to use QRT for preparing a quantum state so that, by measuring the state, we can sample the parameters of the important nodes for the subnetwork with high probability.
To make this algorithm efficient, we never store all parameters of the large original neural network in classical memory but represent them by the amplitude of the quantum state prepared directly from given data.
Conventionally, quantum computation can use QFT to achieve superpolynomial speedups over classical computation for various search problems~\cite{365701,doi:10.1137/S0097539795293172,doi:10.1137/S0097539796300921,9996892}.
By contrast, we make quantum computation applicable to searching in the parameter space of neural networks, by developing QRT to be used in place of QFT\@.

Consequently, our results show that QRT can be used as a fundamental subroutine for QML to accelerate the tasks for the classical neural networks.
A potential drawback may be that our current technique is designed simply for the shallow neural networks with a single hidden layer; however, studies of shallow neural networks capture various essential features of neural networks.
We leave the generalization to deep neural networks for future research, but our development opens a promising route in this direction.
More specifically, if one develops, e.g., a theory to generalize ridgelet transform to the analysis of deep neural networks, our results are expected to offer fundamental techniques for obtaining similar applications to deep neural networks. The significance of our results is to bridge the gap between a theoretical development in analyzing neural networks based on ridgelet transform and the algorithmic techniques in quantum computation for accelerating the machine learning tasks on top of the theoretical development. The essential benefit of ridgelet transform is to provide a closed formula for the weights of the nodes in the hidden layer of a shallow neural network to represent the function to be learned in the over-parameterized (continuous) limit. For deep neural networks, such a closed formula may still be unknown at the moment, but in recent years, significant theoretical progress has been made in the analysis of deep neural networks in the over-parameterized regime. Given such progress, our contributions on the quantum side, if combined with further development in the theoretical analysis of deep neural networks, open a fundamental way to explore the applications in this direction.

\section{\label{sec:ridgelet}Discrete Ridgelet Transform}

\subsection{\label{sec:discrete_ridgelet_transform}Formulation of Discrete Ridgelet Transform}

In this section, we formulate the \newterm{discrete ridgelet transform}.
Then, we also derive a \newterm{Fourier slice theorem} that characterizes our discrete ridgelet transform using Fourier transform.
Although multiple definitions of discrete versions of ridgelet transform have been proposed, none of them has such Fourier expression~\citep{1187351,CARRE20042165,4011958}.
By contrast, the significance of the Fourier slice theorem is that it makes the ridgelet analysis tractable with the well-established techniques for the Fourier transform, which we will use in Sec.~\ref{sec:quantum_ridgelet_transform} for constructing the quantum algorithm as well.

Our formulation assumes the following.
\begin{itemize}
    \item Since computers using a finite number of bits and qubits cannot exactly deal with real number, we use a finite set $\sZ_P\coloneqq\{0,1,\ldots,P-1\}$ in place of the set of real number $\sR$, where $P$ is a precision parameter representing the cardinality of $\sZ_P$. This is conventional in data representation; e.g., for a gray scale image, we may use $8$ bits $\{0,1,\ldots,2^8-1\}$ to represent the intensity of each pixel. In our setting, we can make $P$ larger to achieve better precision in the data representation; e.g., for a more precise representation of the gray scale image, we may use $16$ bits $\{0,1,\ldots,2^{16}-1\}$ in place of the $8$ bits for each pixel. For this improvement of the precision, we may normalize the data by rescaling while the intervals in the discretization are fixed to $1$ for simplicity of presentation. When we write a sum, $\vx$, $\va$, and $\vu$ run over $\sZ_P^D$, and $y$, $b$, and $v$ over $\sZ_P$ unless specified otherwise. Let $\mathcal{F}_D$ and $\mathcal{F}_1$ denote the $D$-dimensional and
    $1$-dimensional discrete Fourier transforms on $\sZ_P^D$ and $\sZ_P$, respectively, i.e.,
    \begin{align}
        \mathcal{F}_D[f](\vu)&\coloneqq
        P^{-\frac{D}{2}}\sum_{\vx}f(\vx)\mathrm{e}^{\frac{-2\pi\mathrm{i}\vu^\top
        \vx}{P}},\\
        \mathcal{F}_1[g](v)&\coloneqq
        P^{-\frac{1}{2}}\sum_{b}r(b)\mathrm{e}^{\frac{-2\pi\mathrm{i}v
        b}{P}}.
    \end{align}
    \item We assume in the following that $P$ is taken as a prime number, considering $\sZ_P$ as a finite field. This is not a restrictive assumption in achieving the better precision since we can take an arbitrarily large prime number as $P$. For example, the maximum of $32$-bit signed integers $P=2^{31}-1$ can be used.
    \item An activation function $g:\sR\to\sR$ is assumed to be normalized as
    \begin{equation}
            \label{eq:g_not_zero}
            \sum_{b}g(b)=0,\quad\|g\|_2^2\coloneqq\sum_{b}|g(b)|^2=1.
    \end{equation}
    Indeed, any non-constant function, such as ReLU, can be used as $g$ with normalization by adding and multiplying appropriate constants.
    \item Corresponding to $g$, we choose a ridgelet function $r:\sR\to\sR$ in such a way that $g$ and $r$ satisfy an admissibility condition
        \begin{equation}
            \label{eq:addmissible}
            C_{g,r}\coloneqq \sum_{v}\mathcal{F}_1[g](v)\overline{\mathcal{F}_1[r](v)}\neq
            0,
    \end{equation}
    where $\overline{\cdots}$ denotes the complex conjugate. We
    also normalize $r$ as $\|r\|_2^2=1$. Note that it is conventional to choose
    $r=g$, leading to $C_{g,g}=\|\mathcal{F}_1[g]\|_2^2=\|g\|_2^2=1$, while our
    setting allows any non-unique choice of $r$ satisfying~\eqref{eq:addmissible} in general.
\end{itemize}

Then, by replacing the integral over $\sR$ in $S$ of~\eqref{eq:continuous} and $R$ of~\eqref{eq:ridgelet} with the finite sum over $\sZ_P$, we correspondingly define the \newterm{discretized neural network} and the \newterm{discrete ridgelet transform} of $f:\sR^D\to\sR$ as, respectively,
\begin{align}
    \label{eq:S}
&\mathcal{S}[w](\vx)\coloneqq P^{-\frac{D}{2}}\sum_{\va,b}w(\va,b)g((\va^\top
\vx-b)\bmod P),\\
    \label{eq:discrete_ridgelet}
&\mathcal{R}[f](\va,b)\coloneqq P^{-\frac{D}{2}}\sum_{\vx}f(\vx)r((\va^\top \vx-b)\bmod P),
\end{align}
where $P^{-\frac{D}{2}}$ represents a normalization constant.
Note that the discretized neural network $\mathcal{S}[w](\vx)$ has discrete parameters $(\va,b)\in\sZ_P^D\times\sZ_P$ for nodes in the hidden layer, but the weights $w(\va,b)\in\sR$ of the nodes are real numbers.

The following theorem, \newterm{Fourier slice theorem}, characterizes the discrete ridgelet transform $R[f]$ in terms of Fourier transforms of $f$ and $h$.
In the case of continuous ridgelet transform, ridgelet analysis can be seen as a form of wavelet analysis in the Radon domain, which combines wavelet and Radon transforms~\citep{Fadili2012}.
The Fourier slice theorem for the continuous ridgelet transform follows from that for continuous Radon transform.
However, problematically, discrete versions of the Radon transforms are nonunique and usually more involved, not necessarily having exact expression in terms of
discrete Fourier transform~\citep{1165108,236530,GOTZ1996711,doi:10.1137/S0097539793256673,862638,doi:10.1137/S003613999732425X,doi:10.1073/pnas.0609228103}.
As a result, the existing definitions of discrete versions of ridgelet transform have no such Fourier expression either~\citep{1187351,CARRE20042165,4011958}.
By contrast, we here show that our formulation of the discrete ridgelet transform $\mathcal{R}[f]$ has the following exact characterization in the Fourier transform.
See Appendix~\ref{sec:A} for proof.

\begin{theorem}[\label{thm:fourier_slice_theorem}Fourier slice theorem for discrete ridgelet transform]
    For any function $f:\sR^D\to\sR$ and any point $\va\in\sZ_P^D,v\in\sZ_P$ in the discretized space, it holds that
    \begin{equation}
        \mathcal{F}_1[\mathcal{R}[f](\va,\cdot)](v)=\mathcal{F}_D[f](v \va\bmod
        P)\;\overline{\mathcal{F}_1[r](v)}.
    \end{equation}
\end{theorem}

\subsection{\label{sec:exact_reconstruction}Exact Representation of Functions as Neural Networks}

Using the discrete ridgelet transform in Sec.~\ref{sec:discrete_ridgelet_transform}, we here show that any function $f$ on the discretized domain has an \newterm{exact representation} in terms of a shallow neural
network with a finite number of parameters in the discretized space.
In the continuous case, any square-integrable function $f$ is represented as the continuous limit of the shallow neural networks, i.e., $f=S[w]$ in~\eqref{eq:continuous}, with the weight given by the ridgelet transform $w\propto R[f]$~\citep{SONODA2017233}.
With discretization, it is nontrivial to show such an exact representation due to finite precision in discretizing the real number.
Nevertheless, we here show that any function $f(\vx)$ for $\vx\in\sZ_P^D$ can be exactly represented as $f(\vx)=\mathcal{S}[w](\vx)$ as well, with the weight given by our formulation of the discrete ridgelet transform $w\propto\mathcal{R}[f]$, which we may call exact representation as a discretized neural network.

In particular, the following theorem shows that any $D$-dimensional real function $f$ on the discrete domain can be exactly represented as a linear combination of non-orthogonal basis functions $\{g((\va^\top\vx-b)\bmod P):(\va,b)\in\sZ_P^D\times\sZ_P\}$ with coefficients given by $\mathcal{R}[f]$.
Due to the non-orthogonality, the coefficients may not be unique, and different choices of the ridgelet function $r$ in~\eqref{eq:discrete_ridgelet} lead to different $\mathcal{R}[f]$
while any of the choices can exactly reconstruct $f$.
The proof is based on the Fourier slice theorem in Theorem~\ref{thm:fourier_slice_theorem}, crucially using the existence of an inverse element for each element of a finite field $\sZ_P$; thus, it is essential to assume that $P$ is a prime number.
See Appendix~\ref{sec:B} for proof.

\begin{theorem}[\label{thm:ridgelet}Exact representation of function as discretized neural network]
    For any function $f:\sR^D\to\sR$ and any point $\vx\in\sZ_P^D$ in the discretized domain, we have \begin{equation}
        f(\vx)=C_{g,r}^{-1}\;\mathcal{S}[\mathcal{R}[f]](\vx),
    \end{equation}
    where $C_{g,r}$ is a constant defined in~\eqref{eq:addmissible}.
\end{theorem}

\section{\label{sec:quantum_ridgelet_transform}Quantum Ridgelet Transform}

In this section, we introduce \newterm{quantum ridgelet transform (QRT)}, an efficient quantum algorithm for implementing the discrete ridgelet transform formulated in Sec.~\ref{sec:ridgelet}.
In various quantum algorithms, we may use quantum Fourier transform (QFT) as a fundamental subroutine.
In addition to QFT, various discrete transforms can be efficiently implemented with quantum computation, such as wavelet transform~\citep{https://doi.org/10.48550/arxiv.quant-ph/9702028,10.5555/645812.670803,938692,10.5555/2011791.2011796,Taha2016,LI2019113,9337176}, Radon transform~\citep{9648027}, fractional Walsh transform~\citep{938691}, Hartley transform~\citep{1328767}, and curvelet transform~\citep{10.1145/1536414.1536469}.
However, the existing discrete versions of ridgelet transform~\citep{1187351,CARRE20042165,4011958} were lacking implementation by quantum computation.
In contrast, our QRT opens a way to use the discrete ridgelet transform as a fundamental subroutine for QML to deal with tasks for classical neural networks.

Basic notions and notations of quantum computation to describe our quantum algorithms are summarized in Appendix~\ref{sec:C}.
In classical computation, we may use $\lceil\log_2(P)\rceil$ bits for representing $\sZ_P$, where $\lceil x\rceil$ denotes the ceiling function, i.e., the smallest integer
that is not smaller than $x$.
The quantum algorithm uses a $\lceil\log_2(P)\rceil$-qubit quantum register for $\sZ_P$.
This quantum register for $\sZ_P$ is represented as a
$2^{\lceil\log_2(P)\rceil}$-dimensional complex vector space $\mathcal{H}_P\coloneqq{(\mathbb{C}^2)}^{\otimes \lceil\log_2(P)\rceil}$.
Using the conventional bra-ket notation, each state in the standard orthonormal basis of the registers is written as a ket (i.e., a vector) $\ket{\vx}\in\mathcal{H}_P^{\otimes D}$ for representing $\vx\in\sZ_P^D$ and $\ket{\va,b}\coloneqq\ket{\va}\otimes\ket{b}\in\mathcal{H}_P^{\otimes D}\otimes\mathcal{H}_P$ for $(\va,b)\in\sZ_P^D\times\sZ_P$, respectively.

The task of QRT is to transform a given unknown quantum state
$\ket{\psi}=\sum_{\vx}\psi(\vx)\ket{\vx}$ into $\mR\ket{\psi}=\sum_{\va,b}\mathcal{R}[\psi](\va,b)\ket{\va,b}$, where $\mR$ is a matrix given by
\begin{align}
    \label{eq:ridgelet_state}
    &\mR\coloneqq P^{-\frac{D}{2}}\sum_{\vx,\va,b}r((\va^\top \vx-b)\bmod
    P)\ket{\va,b}\bra{\vx}.
\end{align}
Under the assumptions in Sec.~\ref{sec:discrete_ridgelet_transform}, $\mR$ becomes an isometry matrix and can be implemented by a quantum circuit for our quantum algorithm, as shown below.
Along with the assumptions in Sec.~\ref{sec:discrete_ridgelet_transform}, we use the following assumption to bound the runtime of our algorithm.
\begin{itemize}
    \item We choose the ridgelet function $r:\sR\to\sR$ in such a way that a
        quantum state representing $r$ by its amplitude $\ket{r}\coloneqq\sum_{y}r(y)\ket{y}$
        can be prepared efficiently in runtime $O(\polylog(P))$. This assumption is
        not restrictive since we can use the quantum algorithm
        by~\citet{https://doi.org/10.48550/arxiv.quant-ph/0208112} to meet the assumption for
        representative choices of $r$ that are integrable efficiently, such as
        ReLU and tanh.
\end{itemize}

Algorithm~\ref{alg:qrt} shows our quantum algorithm for QRT\@.
We construct the algorithm by implementing the discrete Fourier transform in the Fourier slice theorem (Theorem~\ref{thm:ridgelet}) by QFT\@.
QFT applies a unitary matrix representing discrete Fourier transform
\begin{equation}
    \label{eq:QFT}
    \mF_P\coloneqq\sum_{v,b}P^{-\frac{1}{2}}\mathrm{e}^{\frac{-2\pi\mathrm{i}vb}{P}}\ket{v}\bra{b}
\end{equation}
to a given quantum state of $\mathcal{H}_P$ within runtime $O(\polylog(P))$~\citep{doi:10.1142/S0219749904000109}.
The following theorem shows that this speedup is also the case in QRT compared to classical algorithms for computing ridgelet transform.
See Appendix~\ref{sec:C} for proof.

\begin{theorem}[\label{thm:runtime_qrt}Runtime of quantum ridgelet transform]
The runtime of QRT in Algorithm~\ref{alg:qrt} is
\begin{equation}
    O(D\times\polylog(P)).
\end{equation}
\end{theorem}

The advantage of QRT is its linear runtime $O(D)$ in the data dimension $D$, which is exponentially faster than the best existing classical algorithm for ridgelet transform in the $\exp(O(D))$-size space requiring $\exp(O(D))$ runtime.
This advantage is in the same spirit as the QFT being exponentially faster than the corresponding classical algorithm for fast Fourier transform in the spaces of the same size.
In Sec.~\ref{sec:application}, we will further clarify that QRT has an application to accelerate the task of finding the winning ticket of neural networks.

\begin{algorithm}[t]
  \caption{\label{alg:qrt}Quantum ridgelet transform (QRT).}
  \begin{algorithmic}[1]
      \REQUIRE{A given input state $\ket{\psi}=\sum_{\vx}\psi(\vx)\Ket{\vx}$, the ridgelet function $r$ satisfying the assumptions in Sec.~\ref{sec:quantum_ridgelet_transform}.}
      \ENSURE{Output $\mR\ket{\psi}=\sum_{\va,b}\mathcal{R}[\psi](\va,b)\ket{\va,b}$ in~\eqref{eq:ridgelet_state} within runtime $O(D\times\polylog(P))$ as in Theorem~\ref{thm:runtime_qrt}.}
      \STATE{Given $\ket{\psi}$, add an auxiliary register $\mathcal{H}_P$ prepared in $\sum_{b}r(b)\ket{b}$, to obtain $\sum_{\vx,b}\psi(\vx)\ket{\vx}\otimes r(b)\ket{b}$.
  }
    \STATE{Apply $D$-dimentional QFT $\mF_P^{\otimes D}$ and $1$-dimensional inverse QFT $\mF_P^\dag$ to the first and second quantum registers, respectively, to transform $\sum_{\vx,b}\psi(\vx)\ket{\vx}\otimes r(b)\ket{b}$ into
    \begin{align}
        &\sum_{\va^\prime\in\sZ_P^D}\sum_{v\in\sZ_P}\mathcal{F}_D[\psi](\va^\prime)\ket{\va^\prime}\otimes \overline{\mathcal{F}_1[r](v)}\ket{v}.
    \end{align}
}
    \STATE{Perform arithmetics $\ket{\va^\prime}\mapsto\ket{v^{-1}\va^\prime\bmod P}$ on the first register by controlled gates that are controlled by the state $\ket{v}$ of the second, to obtain
    \begin{align}
        &\hspace{-1em}\sum_{v\in\mathbb{Z}_P\setminus\{0\}}\sum_{\va^\prime\in\sZ_P^D}\mathcal{F}_D[\psi](\va^\prime)\Ket{v^{-1}\va^\prime\bmod P}\otimes\overline{\mathcal{F}_1[r](v)}\ket{v}\nonumber\\
        &\hspace{-1em}=\sum_{\va,v}\mathcal{F}_D[\psi](v \va\bmod P)\overline{\mathcal{F}_1[r](v)}\ket{\va}\otimes\Ket{v},
    \end{align}
    where $v^{-1}\in\mathbb{Z}_P$ is the inverse of $v$ in the finite field $\mathbb{Z}_P$, $\va\coloneqq v^{-1}\va^\prime$, and $\bmod P$ for $\sZ_P^D$ is taken element-wise.}
    \STATE{Apply inverse QFT $\mF_P^\dag$ to the second quantum register, which yields $\mR\ket{\psi}$ due to Theorem~\ref{thm:fourier_slice_theorem}.}
    \STATE{\textbf{Return } $\mR\ket{\psi}$.}
  \end{algorithmic}
\end{algorithm}

\section{\label{sec:application}Application of Quantum Ridgelet Transform to Lottery Ticket Hypothesis}

\subsection{\label{sec:setting}Setting for Winning Ticket of Neural Networks}

In this section, as an application of quantum ridgelet transform (QRT) in Sec.~\ref{sec:quantum_ridgelet_transform},
we propose an algorithm for finding a sparse subnetwork approximating a large neural network efficiently by quantum computation, based on the \newterm{lottery ticket hypothesis} on neural networks.
The lottery ticket hypothesis by~\citet{frankle2018the} claims that a randomly-initialized fully-connected neural network contains a subnetwork that is initialized in such a way that, when trained in isolation, it can match the accuracy of the original network after training for at most the same number of iterations.
This hypothesis has been confirmed numerically in various settings.
The theoretical analysis of deep neural networks is inevitably hard in general, and studies of shallow neural networks are also important for capturing essences of neural networks, which we here consider in a setting of regression from given data as described in the following.

For $D\in\{1,2,\ldots\}$ with a fixed prime number $P$, we consider a family of problems to approximate an unknown function $f^{(D)}:\sR^D\to\sR$ by a shallow neural network, i.e.,
\begin{equation}
\label{eq:neural_network}
    \hat{f}^{(D)}(\vx)\coloneqq\sum_{n=1}^{N}w_n g((\va_n^\top \vx-b_n)\bmod P).
\end{equation}
Let $p^{(D)}_\mathrm{data}$ be a probability mass function for the input data, which is assumed to be supported on $\sZ_P^D$.
Suppose that we are given $M$ input-output pairs of examples $(\vx_1,y_1=f(\vx_1)),\ldots,(\vx_{M},y_{M}=f(\vx_{M}))\in\sZ_P^D\times\sR$.
Let $\hat{p}^{(D)}_\mathrm{data}$ denote the empirical distribution of $\vx_1,\ldots,\vx_M$.
Given $\epsilon>0$, we will analyze empirical risk minimization~\citep{bach2021learning}, i.e., minimization of the empirical risk $\sum_{\vx} \hat{p}^{(D)}_\mathrm{data}(\vx)|f^{(D)}(\vx)-\hat{f}^{(D)}(\vx)|^2$ to $O(\epsilon)$.
If obvious, we may omit $D$ in superscripts; e.g., we may write $f^{(D)}$ as $f$.

The setting of our analysis, along with the assumptions in Sec.~\ref{sec:quantum_ridgelet_transform}, is as follows.
Based on the exact representation of  $f(\vx)$ in terms of the neural network $\mathcal{S}[w](\vx)$ in Theorem~\ref{thm:ridgelet},
we can approximate $f$ by a neural network
\begin{equation}
\label{eq:f_approximation}
f(\vx)\!\approx\!\mathcal{S}[w_\lambda^\ast](\vx)\!=\!\!\sum_{\va,b}P^{-\frac{D}{2}}w_\lambda^\ast(\va,b)g((\va^\top \vx-b)\bmod P),
\end{equation}
where $w_\lambda^\ast$ is the optimal solution of the ridge regression with the empirical distribution, i.e.,
\begin{align}
    \label{eq:minimization}
    &w_\lambda^\ast(\va,b)\coloneqq\argmin_{w}\{\tilde{J}(w)\},
\end{align}
$\tilde{J}(w)\coloneqq J(w)+\lambda\Omega(w)$,
$J(w)\coloneqq\sum_{\vx}\hat{p}_\mathrm{data}(\vx)|f(\vx)-\mathcal{S}[w](\vx)|^2$,
$\Omega(w)\coloneqq\|P^{-\frac{D}{2}}w\|_2^2=\sum_{\va,b}|P^{-\frac{D}{2}}w(\va,b)|^2$,
and $\lambda>0$ is a hyperparameter for regularization.
Learning a general class of function $f^{(D)}$ on $\sZ_P^D$ would be inevitably demanding as its representation would require $\exp(O(D))$ parameters to specify the values $f^{(D)}(\vx)$ for all $\vx\in\sZ_P^D$ in the worst case.
We have shown such a general representation in Theorem~\ref{thm:ridgelet}.
By contrast, our goal here is to achieve the approximation feasibly with much fewer parameters, using a subnetwork of the large original network $\mathcal{S}[w_\lambda^\ast](\vx)$.

To this goal, recall that it is conventional in statistical learning theory to consider a reasonably restricted class of functions, e.g., those with bounded norms~\citep{bach2021learning}; correspondingly, we work on a setting where the norm $\|P^{-\frac{D}{2}}w_\lambda^\ast\|_1\coloneqq\sum_{\va,b}|P^{-\frac{D}{2}}w_\lambda^\ast(\va,b)|$ of the weights for representing $f^{(D)}$ should be bounded even on the large scales $D\to\infty$.
In particular, let $((\va_j,b_j)\in\sZ_P^D\times\sZ_P:j\in\{1,\ldots,P^{D+1}\})$ denote a sequence of parameters of all nodes in the hidden layer of $\mathcal{S}[w_\lambda^\ast](\vx)$ aligned in the descending order of $w_\lambda^\ast$, i.e., $|w_\lambda^\ast(\va_1,b_1)|\geqq|w_\lambda^\ast(\va_2,b_2)|\geqq\cdots$.
These nodes of $\mathcal{S}[w_\lambda^\ast]$ are ordered in the same way; i.e., the weight of the $j$th node is $w_\lambda^\ast(\va_j,b_j)$.
Then, we assume the following.
\begin{itemize}
    \item Following the convention of assumptions in the previous works by, e.g.,~\citet{donoho1993unconditional,hayakawa2020minimax}, we assume that there exist constants $\alpha,\beta>0$ such that it holds uniformly for any $D$ and $j$ that
\begin{equation}
    \label{eq:decay}
    |P^{-\frac{D}{2}}w_\lambda^\ast(\va_j,b_j)|\leqq\alpha j^{-(1+\beta)},
\end{equation}
which specifies the decay rate for large $j$, leading to $\|P^{-\frac{D}{2}}w_\lambda^\ast\|_1\leqq\sum_{j=1}^{\infty}\alpha j^{-(1+\beta)}\leqq\alpha+\int_1^\infty\frac{\alpha}{x^{1+\beta}}dx=\alpha+\frac{\alpha}{\beta}<\infty$.
We also write the $L^2$ norm as
\begin{equation}
\label{eq:gamma}
    \gamma\coloneqq\|P^{-\frac{D}{2}}w_\lambda^\ast\|_2^2,
\end{equation}
which is upper bounded by
$\gamma\leqq\sum_{j=1}^{\infty}\alpha^2 j^{-2(1+\beta)}\leqq\alpha^2+\int_1^\infty\frac{\alpha^2}{x^{2(1+\beta)}}dx=\alpha^2+\frac{\alpha^2}{1+2\beta}$.
Functions $(f^{(D)}:D=1,2,\ldots)$ satisfying~\eqref{eq:decay} are called $(\alpha,\beta)$-class functions, which are to be learned in our setting.
\end{itemize}

For given $\epsilon>0$,
our analysis focuses on the task of finding a sparse representation $\hat{f}$ to approximate $\mathcal{S}[w_\lambda^\ast]$ up to $\epsilon$, i.e.,
\begin{equation}
    \label{eq:S_hat_f_error}
    \sum_{\vx}\hat{p}_\mathrm{data}(\vx)|\mathcal{S}[w_\lambda^\ast](\vx)-\hat{f}(\vx)|^2=O(\epsilon),
\end{equation}
with keeping the number of nodes $N$ in the hidden layer of $\hat{f}$ in~\eqref{eq:neural_network} as small as possible.
Following the conventional prescription in the statistical learning theory~\citep{bach2021learning}, we assume that we appropriately choose
\begin{equation}
    \label{eq:lambda}
\lambda\approx\poly(\epsilon),
\end{equation}
so that~\eqref{eq:S_hat_f_error} leads to $\sum_{\vx}\hat{p}_\mathrm{data}(\vx)|f(\vx)-\hat{f}(\vx)|^2=O(\epsilon)$, achieving the empirical risk minimization with $\hat{f}$.
Note that if we fix $D$, then for any $f$, we may be able to find sufficiently large $\alpha$ and small $\beta$ to meet~\eqref{eq:decay}, but our analysis will show that assuming smaller $\alpha$ and larger $\beta$ guarantees smaller $N$ to achieve~\eqref{eq:S_hat_f_error} for the $(\alpha,\beta)$-class functions for arbitrary $D$.

\subsection{\label{sec:winning}Quantum Algorithm for Sampling from Optimized Probability Distribution}

We here construct a quantum algorithm for sampling from an \newterm{optimized probability distribution} (defined later as $p_{\lambda,\Delta}^\ast(\va,b)$ in~\eqref{eq:optimized_distribution}) of parameters of nodes in the hidden layer of the large original network $\mathcal{S}[w_\lambda^\ast]$ in~\eqref{eq:f_approximation}, which we can use for efficiently finding a sparse subnetwork of $\mathcal{S}[w_\lambda^\ast]$ to approximate $f$ well.
The original network $\mathcal{S}[w_\lambda^\ast]$ has $\exp(O(D))$ nodes to represent any function $f$, as with Theorem~\ref{thm:ridgelet}.
By contrast, studies on the lottery ticket hypothesis provide numerical evidences that $f$ in practice can usually be approximated by a sparse subnetwork with much fewer parameters.
One existing way to find such a subnetwork is to train the overall large network and then perform masking to eliminate the low-weight nodes while keeping those with higher weights~\citep{frankle2018the}.
However, this approach is inefficient since one needs large-scale optimization to train the large original network before the pruning.
Then, more recent studies by~\citet{lee2018snip,NEURIPS2019_1113d7a7,Ramanujan_2020_CVPR,pmlr-v119-malach20a,NEURIPS2020_1e949147,NEURIPS2020_1b742ae2,NEURIPS2020_46a4378f,NEURIPS2020_ad1f8bb9,Wang2020Picking,Wang_Zhang_Xie_Zhou_Su_Zhang_Hu_2020,frankle2021pruning,chen2022peekaboo} have suggested that one should be able to find the subnetwork only by pruning the initial network directly, even without the optimization for training.
Still, to perform this pruning appropriately, one needs to perform a large-scale search for the subnetwork within the parameter space of the large original neural network.
As $D$ increases, it would become infeasible to deal with the large original network for training or searching as long as we use the existing methods based on classical computation.

To address this problem, our key idea is to represent the weights of the $\exp(O(D))$ nodes in the hidden layer of the neural network $\mathcal{S}[w_\lambda^\ast]$ efficiently as the amplitude of quantum state of only $O(D)$ qubits.
Roughly speaking, as in Theorem~\ref{thm:ridgelet}, these weights can be given by the discrete ridgelet transform $\mathcal{R}[f]$ of $f$, which is implementable efficiently by QRT\@.
In particular, if we initially have a quantum state $\ket{f}=\sum_{\vx}f(\vx)\ket{\vx}$, then QRT of $\ket{f}$ can prepare $\mR\ket{f}=\sum_{\va,b}\mathcal{R}[f](\va,b)\ket{\va,b}$ in time $\widetilde{O}(D)$ as shown in Theorem~\ref{thm:runtime_qrt}, where $\widetilde{O}$ may ignore polylogarithmic factors.
A measurement of this quantum state $\mR\ket{f}$ in basis $\{\ket{\va,b}\}$ provides a measurement outcome $(\va,b)$ sampled from a probability distribution proportional to the square of the amplitude, i.e., $|\mathcal{R}[f](\va,b)|^2$.
In this way, we can find parameter $(\va,b)$ for a node with large $|\mathcal{R}[f](\va,b)|^2$ with high probability, in runtime $\widetilde{O}(D)$ per sampling.
The state is corrupted by the measurement, and to perform the sampling $N$ times, we repeat the preparation and measurement $N$ times.
To sample $(\va,b)$ for all high-weight nodes with high probability in a theoretically guaranteed way,
for $\Delta>0$,
we introduce an \newterm{optimized probability distribution}
\begin{equation}
    \label{eq:optimized_distribution}
    p_{\lambda,\Delta}^\ast(\va,b)\coloneqq\frac{1}{Z}\frac{|P^{-\frac{D}{2}}w_\lambda^\ast(\va,b)|^2}{|P^{-\frac{D}{2}}w_\lambda^\ast(\va,b)|^2+\Delta},
\end{equation}
where $Z$ is a constant for normalization $\sum_{\va,b}p_{\lambda,\Delta}^\ast(\va,b)=1$.
Appropriate $\Delta$ for our task in~\eqref{eq:S_hat_f_error} will be specified later in Theorem~\ref{thm:winning}.
To sample from $p_{\lambda,\Delta}^\ast$, we prepare
\begin{equation}
    \ket{p_{\lambda,\Delta}^\ast}\coloneqq\frac{1}{\sqrt{Z}}\sum_{\va,b}\frac{P^{-\frac{D}{2}}w_\lambda^\ast(\va,b)}{\sqrt{|P^{-\frac{D}{2}}w_\lambda^\ast(\va,b)|^2+\Delta}}\ket{\va,b},
\end{equation}
followed by performing the measurement of $\ket{p_{\lambda,\Delta}^\ast}$ in the same way as described above.

To apply the above idea to practical tasks of machine learning, it is important to deal with a conventional situation where the data (e.g., examples of input-output pairs of $f$) are given by classical bit strings rather than quantum states; thus, a critical issue for our algorithm should be how to give such a quantum state from classical data.
To achieve the overall task using QRT, we first need to input the classical data by converting the data into a quantum state, then apply QRT to the quantum state, and finally perform a measurement to obtain a classical output from the quantum state.
We also remark that, in some other proposals of QML, some ``quantum'' data may be assumed to be given by quantum states, e.g., as a result of another quantum algorithm or physical process, and QRT is also potentially useful in such a quantum setting.
But significantly, our algorithm here avoids such an assumption by explicitly clarifying how to prepare an input quantum state from the given classical examples in the task of finding the winning ticket of neural networks, as shown below.

In particular, we explicitly construct an input model for our quantum algorithm by quantum circuit as follows.
\begin{itemize}
    \item As an input, our algorithm uses quantum circuits to prepare $\ket{\hat{p}_\mathrm{data}}=\sum_{\vx}\sqrt{\hat{p}_\mathrm{data}(\vx)}\ket{\vx}$ and $\ket{\psi_\mathrm{in}}\propto\sum_{\vx}\hat{p}_\mathrm{data}(\vx)f(\vx)\ket{\vx}$.
Regarding $\ket{\hat{p}_\mathrm{data}}$,
if we prepare $\ket{\hat{p}_\mathrm{data}}$ and  measure it in basis $\{\ket{\vx}\}$,
we can randomly sample $\vx$ according to the empirical distribution $\hat{p}_\mathrm{data}(\vx)$.
For a classical algorithm, sampling from $\hat{p}_\mathrm{data}$ can be easily realized in time $O(D\polylog(M))$ over $M$ examples of $D$-dimensional input, by sampling $m\in\{1,\ldots,M\}$ from the uniform distribution over $O(\log(M))$ bits and outputting $\vx_m\in\sZ_P^D$ out of $\vx_1,\ldots,\vx_{M}$ stored in random access memory (RAM).
As for the quantum algorithm, the preparation of $\ket{\hat{p}_\mathrm{data}}$ with maintaining quantum superposition may be more technical.
But we show that this preparation is also implementable in time $O(D \polylog(M))$, by storing the $M$ examples upon collecting them in a sparse binary-tree data structure~\citep{kerenidis_et_al:LIPIcs:2017:8154} with quantum RAM (QRAM)~\citep{PhysRevA.78.052310,PhysRevLett.100.160501}, where QRAM is implemented explicitly as a parallelized quantum circuit of depth $O(\polylog(M))$~\citep{8962352,PRXQuantum.2.020311}.
Using the same data structure, we also show that the preparation of $\ket{\psi_\mathrm{in}}$ is implemented within the same runtime $O(D \polylog(M))$.
As a whole, our assumption is to store the $M$ examples in these data structures upon collecting them, so that each preparation of $\ket{\hat{p}_\mathrm{data}}$ and $\ket{\psi_\mathrm{in}}$ has runtime
$O(D \polylog(M))=\widetilde{O}(D)$.
See Appendix~\ref{sec:D} for detail.
\end{itemize}

The following theorem shows that we have a quantum algorithm that can prepare and measure $\ket{p_{\lambda,\Delta}^\ast}$ to sample from $p_{\lambda,\Delta}^\ast$ within a linear runtime $\widetilde{O}(\frac{D}{\lambda\Delta})$.
Our algorithm is based on the analytical formula for the solution of ridge regression in~\eqref{eq:minimization}; in particular, with $r=g$, the formula leads to
\begin{equation}
\label{eq:state_W}
    \ket{p_{\lambda,\Delta}^\ast}\propto{\Big(\mW_\lambda+\frac{\Delta}{\gamma}\mI\Big)}^{-\frac{1}{2}}{\Big(\mR\hat{\mP}_\mathrm{data}\mR^\top+\lambda\mI\Big)}^{-1}\mR\ket{\psi_\mathrm{in}},
\end{equation}
where $\mW_\lambda\coloneqq\gamma^{-1}\sum_{\va,b}|P^{-\frac{D}{2}}w_\lambda^\ast(\va,b)|^2\ket{\va,b}\bra{\va,b}$ and $\hat{\mP}_\mathrm{data}\coloneqq\sum_{\vx}\hat{p}_\mathrm{data}(\vx)\ket{\vx}\bra{\vx}$.
The algorithm is also explicitly presented as Algorithm~\ref{alg:sampling} in Appendix~\ref{sec:D}.
See Appendix~\ref{sec:D} for proof of its runtime as well.

\begin{theorem}[\label{thm:quantum_sampling}Runtime of quantum algorithm for sampling from optimized probability distribution]
    Given $\lambda,\Delta>0$, for any $D$, a quantum algorithm can prepare and measure $\ket{p_{\lambda,\Delta}^\ast}$ to sample from $p_{\lambda,\Delta}^\ast(\va,b)$ within runtime $\widetilde{O}\left(\frac{D}{\lambda\Delta}\times\gamma\right)$ per sampling, where $\gamma$ is a constant in~\eqref{eq:gamma}.
\end{theorem}

Remarkably, our construction of the quantum algorithm for Theorem~\ref{thm:quantum_sampling} is based on two significant technical contributions.
First,
estimation of classical description of $\ket{p_{\lambda,\Delta}^\ast}$ would need $\exp(O(D))$ runtime and may cancel out the advantage of QML~\citep{aaronson2015read}, but we avoid such slowdown.
In particular, our quantum algorithm prepares $\ket{p_{\lambda,\Delta}^\ast}$ directly from the $M$ examples and then measure it to obtain parameter $(\va,b)$ for a high-weight node of $\mathcal{S}[w_\lambda^\ast]$ per single preparation and measurement.
In this way, we circumvent the costly process of expectation-value estimation throughout our algorithm.
Second, we develop a technique for implementing the inverses of $\exp(O(D))\times\exp(O(D))$ matrices in~\eqref{eq:state_W} with quantum computation efficiently, yet without imposing restrictive assumptions.
In particular, the inverses of $\exp(O(D))\times\exp(O(D))$ matrices are hard to compute in classical computation, and conventional techniques in QML have required sparsity or low-rankness of the matrices to implement matrix inversion with large quantum speedups~\citep{10.1145/3313276.3316366}.
More recent quantum-inspired classical algorithms also require the low-rank assumption~\citep{10.1145/3313276.3316310}.
However, the matrices to be inverted in our algorithm are not necessarily sparse or low-rank, and thus imposing such assumptions would limit the applicability of QML\@.
By contrast, we avoid imposing these assumptions by directly clarifying the quantum circuits for implementing these matrices efficiently with QRT\@.
In the existing research, this type of technique for avoiding the sparsity and low-rankness assumptions in QML was established only for Fourier transform~\citep{NEURIPS2020_9ddb9dd5,https://doi.org/10.48550/arxiv.2106.09028}.
Our development discovers wide applicability of such techniques even to a broader class of transforms including ridgelet transform.

\subsection{\label{sec:finding}Quantum Algorithm for Finding Winning Ticket of Neural Networks and Performance Analysis}

Using the quantum algorithm for sampling from $p_{\lambda,\Delta}^\ast(\va,b)$ in Theorem~\ref{thm:quantum_sampling}, we describe \newterm{an algorithm for finding a winning ticket}, i.e., a sparse trainable subnetwork of the large original network $\mathcal{S}[w_\lambda^\ast]$ in~\eqref{eq:f_approximation} for approximating $f$.
We also analyze its performance with theoretical guarantee.

We here describe our algorithm for finding a winning ticket, which is also explicitly presented as Algorithm~\ref{alg:winning} in Appendix~\ref{sec:E}.
In our algorithm, we repeat the sampling from  $p_{\lambda,\Delta}^\ast(\va,b)$ in total $N$ times by the quantum algorithm in Theorem~\ref{thm:quantum_sampling}, where $N$ is given later in~\eqref{eq:Delta_bound}.
Letting  $\hat{\sW}$ denote the set of sampled parameters in these $N$ repetitions, we approximate $\mathcal{S}[w_\lambda^\ast]$ by the subnetwork
\begin{equation}
    \label{eq:sparse_subnetwork}
    \!\!\!\mathcal{S}[w_\lambda^\ast]\!\approx\!\hat{f}(\vx)\!=\!\!\!\!\!\!\sum_{(\va,b)\in\hat{\sW}}\hat{w}^\ast(\va,b)g((\va^\top \vx-b)\bmod P),
\end{equation}
where we write $\hat{\vw}^\ast\coloneqq(\hat{w}^\ast(\va,b)\in\sR:(\va,b)\in\hat{\sW})$.
Each sampling provides parameter $(\va,b)\in\hat{\sW}$ of each node in the hidden layer of this subnetwork but not the value of $\hat{w}^\ast(\va,b)$.
Once we fix $\{g((\va^\top \vx-b)\bmod P):(\va,b)\in\hat{\sW}\}$ of the subnetwork, we then train $\hat{\vw}^\ast$ efficiently by the established classical algorithms for convex optimization such as stochastic gradient descent (SGD)~\citep{pmlr-v99-harvey19a}, using the $M$ examples.
Thus, the quantum speedup is not cancelled out throughout the learning including the training of $\hat{\vw}^\ast$.
In this way, we achieve our task~\eqref{eq:S_hat_f_error} with trainability.

The following theorem guarantees $\Delta$ and $N$ required for achieving our task~\eqref{eq:S_hat_f_error}.
By combining Theorems~\ref{thm:quantum_sampling} and~\ref{thm:winning}
with~\eqref{eq:lambda} in our setting,
the overall runtime of our algorithm is
\begin{equation}
    \widetilde{O}\Big(N\times \frac{D}{\lambda\Delta}\gamma\Big)=\widetilde{O}\Big(\frac{D}{\lambda\epsilon^{1+\nicefrac{2}{\beta}}}\Big)=\widetilde{O}\Big(D\times\poly\Big(\frac{1}{\epsilon}\Big)\Big),
\end{equation}
dominated by the $N$ repetitions of the sampling in Theorem~\ref{thm:quantum_sampling}.
A comparison with classical algorithms analogous to our sampling-based approach is made in Sec.~\ref{sec:advantage}.
See Appendix~\ref{sec:E} for proof.

\begin{theorem}[\label{thm:winning}Bounds for finding winning ticket of neural networks]
Given $\epsilon,\delta>0$,
there exist $\Delta$ and $N$ satisfying
\begin{align}
\label{eq:Delta_bound}
\Delta=\Omega\left(\epsilon^{1+\frac{1}{\beta}}\right),\;N=O\left( \epsilon^{-\frac{1}{2\beta}}\log\left(\epsilon^{-1}\delta^{-1}\right)\right),
\end{align}
such that the algorithm described above returns a subnetwork $\hat{f}$ of the neural network $\mathcal{S}[w_\lambda^\ast]$ with the number of nodes in the hidden layer of $\hat{f}$ smaller than $N$, and $\hat{f}$ achieves the task of approximating $\mathcal{S}[w_\lambda^\ast]$ to $O(\epsilon)$ in~\eqref{eq:S_hat_f_error} with high probability greater than $1-\delta$.
\end{theorem}
\subsection{\label{sec:advantage}Advantage of Using Quantum Ridgelet Transform}

\begin{figure}[t]
  \centering
  \includegraphics[width=\columnwidth]{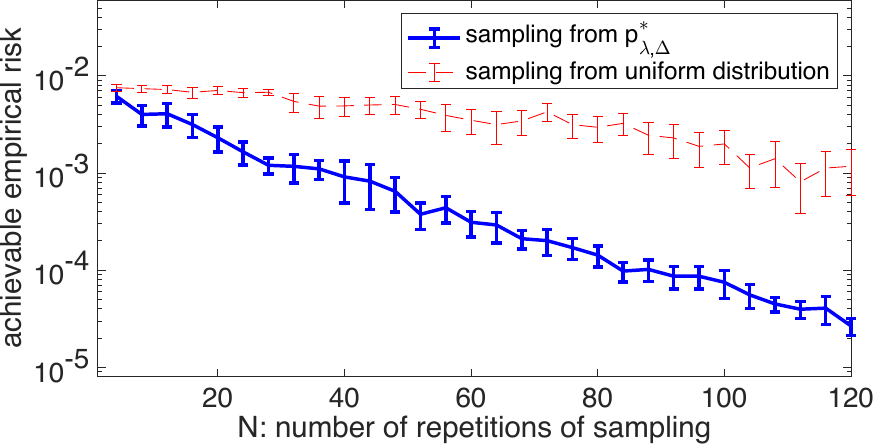}
  \vspace{-5mm}
  \caption{\label{fig:advantage}The empirical risks achievable with the subnetworks of the large original neural network found by $N$ repetitions of sampling from the optimized distribution $p_{\lambda,\Delta}^\ast(\va,b)$ in~\eqref{eq:optimized_distribution} via our algorithm in Theorem~\ref{thm:winning} (blue thick line), and that from the uniform distribution via random features (red dashed line). Each line represents the average over $20$ executions of the algorithms, while each error bar represents the unbiased estimation of the standard deviation for these executions. The advantage of using our algorithm over the simple application of the random features can be order of magnitude in terms of the empirical risk in this regime.}
\end{figure}

We numerically demonstrate \newterm{the advantage of our algorithm for winning the lottery ticket} of the neural network in Theorem~\ref{thm:winning}.
For fair comparison between our algorithm and a similar approach for classical algorithms, recall that the essential idea of our algorithm is to avoid computationally hard optimization of the initial neural network by sampling the nodes in its hidden layer to decide the basis functions $\{g((\va^\top \vx-b)\bmod P):(\va,b)\in\hat{\sW}\}$ in~\eqref{eq:sparse_subnetwork}, followed by efficiently training their coefficients $\boldsymbol{\hat{w}^\ast}$ via convex optimization.
This idea can be regarded as a generalized form of random features by~\citet{R2,R3}, where one randomly samples feature maps, i.e., $\{g((\va^\top \vx-b)\bmod P)\}$, and then find their coefficients by convex optimization.
The difference is that we use the optimized distribution $p_{\lambda,\Delta}^\ast(\va,b)$ depending on $f$ and $\hat{p}_\mathrm{data}$, but the random features conventionally performs sampling from a distribution independent of $f$ and $\hat{p}_\mathrm{data}$, e.g., a uniform distribution $p_{\mathrm{uniform}}(\va,b)\coloneqq\frac{1}{P^{D+1}}$.
Note that a more recent work by~\citet{B1} has proposed to sample optimized random features depending on the data distribution (but still not on $f$ itself), and this sampling is also efficiently achievable with a quantum algorithm shown by~\citet{NEURIPS2020_9ddb9dd5,https://doi.org/10.48550/arxiv.2106.09028}; however, without QRT in this work, it would not be straightforward to apply such techniques to neural networks.
Despite this difference, a quantitative advantage of our algorithm over the random features would be still unclear without numerical simulation, due to the non-orthogonality of the basis functions $\{g((\va^\top \vx-b)\bmod P)\}$ used for the neural network.

We numerically show that our algorithm can find a subnetwork achieving a significantly better empirical risk than that obtained from the random features, as illustrated in Fig.~\ref{fig:advantage}.
In our numerical experiment, choosing $D=1$ and $P=127$, we set the function $f$ to be learned as a sine function, the empirical distribution as the uniform distribution $\hat{p}_\mathrm{data}(\vx)=\frac{1}{P^D}$, and the activation function $g$ and the ridgelet function $r$ as ReLU\@.
The sampling from $p_{\lambda,\Delta}^\ast$ was classically simulated via rejection sampling, and convex optimization of $\hat{w}^\ast$ was solved by MOSEK~\citep{mosek} and YALMIP~\citep{Lofberg2004}.
For $N\leqq 120$, we plotted the achievable empirical risk with $N$ repetitions of the sampling from $p_{\lambda,\Delta}^\ast$ and that from $p_\mathrm{uniform}$.
The advantage of our algorithm in finding a sparse trainable subnetwork over the random features can be order of magnitude in terms of the empirical risk achievable by the subnetwork.
We also note that a similar advantage can also be obtained in the case where we choose the activation function as a sigmoid function (tanh) rather than ReLU, which supports our results further.
See Appendix~\ref{sec:F} for detail.

We emphasize that our classical simulation using rejection sampling of $p_{\lambda,\Delta}^\ast$ is not scalable as $D$ increases; by contrast, our results make it possible to take the advantage in higher dimension $D\gg 1$ if we can use quantum computation for accelerating the task.
Our main contributions are the development of algorithmic techniques for QML and the derivation of the bounds on the performances of our algorithms with theoretical guarantees, which are not heuristic.
We may be able to perform the numerical experiment only for small $D$ at the moment because the numerical experiment for higher dimensions is computationally hard as long as we use classical computation to simulate quantum computation.
However, together with our theoretical analysis, the overall results lay a solid foundation for developing fault-tolerant quantum computers to demonstrate the advantage of our algorithms for the higher dimensions in quantum experiments in the future.

\section{Conclusion}
We have formulated the discrete ridgelet transform that can be characterized via Fourier slice theorem and can represent any function exactly in the discretized domain. Furthermore, as a fundamental subroutine for quantum machine learning (QML), we have constructed quantum ridgelet transform (QRT), a quantum algorithm for applying $D$-dimensional discrete ridgelet transform to a quantum state efficiently in linear time $\widetilde{O}(D)$.
We have also clarified an application of QRT for finding a sparse trainable subnetwork of a large-scale neural network to approximate a function to be learned, opening an efficient way to demonstrate lottery ticket hypothesis.
These results discover a promising use of QML to accelerate the tasks for classical neural networks.
Also from a broader perspective, our quantum algorithms may need a fault-tolerant quantum computer that is actively under development, and our achievement lays a solid theoretical foundation for further hardware development and social implementation toward realizing quantum computation.

\section*{Acknowledgements}
Hayata Yamasaki was supported by JSPS Overseas Research Fellowship, JST PRESTO Grant Number JPMJPR201A and MEXT Quantum Leap Flagship Program (MEXT QLEAP) JPMXS0118069605, JPMXS0120351339\@.
Sathyawageeswar Subramanian was supported by Royal Commission for the Exhibition of 1851 Research Fellowship in Science and Engineering.
Sho Sonoda was supported by JST PRESTO Grant Number JPMJPR2125\@.

\bibliographystyle{icml2023}
\bibliography{citation_bibtex}

\begin{thebibliography}{88}
\providecommand{\natexlab}[1]{#1}
\providecommand{\url}[1]{\texttt{#1}}
\expandafter\ifx\csname urlstyle\endcsname\relax
  \providecommand{\doi}[1]{doi: #1}\else
  \providecommand{\doi}{doi: \begingroup \urlstyle{rm}\Url}\fi

\bibitem[Aaronson(2015)]{aaronson2015read}
Aaronson, S.
\newblock Read the fine print.
\newblock \emph{Nature Physics}, 11\penalty0 (4):\penalty0 291, 2015.
\newblock URL \url{https://www.nature.com/articles/nphys3272}.

\bibitem[Ambainis(2012)]{ambainis:LIPIcs:2012:3426}
Ambainis, A.
\newblock {Variable time amplitude amplification and quantum algorithms for
  linear algebra problems}.
\newblock In D{\"u}rr, C. and Wilke, T. (eds.), \emph{29th International
  Symposium on Theoretical Aspects of Computer Science (STACS 2012)}, volume~14
  of \emph{Leibniz International Proceedings in Informatics (LIPIcs)}, pp.\
  636--647, Dagstuhl, Germany, 2012. Schloss Dagstuhl--Leibniz-Zentrum fuer
  Informatik.
\newblock ISBN 978-3-939897-35-4.
\newblock \doi{10.4230/LIPIcs.STACS.2012.636}.
\newblock URL \url{http://drops.dagstuhl.de/opus/volltexte/2012/3426}.

\bibitem[ApS(2022)]{mosek}
ApS, M.
\newblock \emph{The MOSEK optimization toolbox for MATLAB manual. Version
  9.3.21.}, 2022.
\newblock URL \url{http://docs.mosek.com/9.3/toolbox/index.html}.

\bibitem[Arg\"{u}ello(2009)]{10.5555/2011791.2011796}
Arg\"{u}ello, F.
\newblock Quantum wavelet transforms of any order.
\newblock \emph{Quantum Info. Comput.}, 9\penalty0 (5):\penalty0 414–422, may
  2009.
\newblock ISSN 1533-7146.
\newblock URL \url{https://dl.acm.org/doi/abs/10.5555/2011791.2011796}.

\bibitem[Bach(2017)]{B1}
Bach, F.
\newblock On the equivalence between kernel quadrature rules and random feature
  expansions.
\newblock \emph{Journal of Machine Learning Research}, 18\penalty0
  (21):\penalty0 1--38, 2017.
\newblock URL \url{http://jmlr.org/papers/v18/15-178.html}.

\bibitem[Bach(2021)]{bach2021learning}
Bach, F.
\newblock \emph{Learning Theory from First Principles}.
\newblock 2021.
\newblock URL \url{https://www.di.ens.fr/%7Efbach/ltfp_book.pdf}.

\bibitem[Barron(1993)]{256500}
Barron, A.
\newblock Universal approximation bounds for superpositions of a sigmoidal
  function.
\newblock \emph{IEEE Transactions on Information Theory}, 39\penalty0
  (3):\penalty0 930--945, 1993.
\newblock \doi{10.1109/18.256500}.
\newblock URL \url{https://ieeexplore.ieee.org/document/256500/}.

\bibitem[Bernstein \& Vazirani(1997)Bernstein and
  Vazirani]{doi:10.1137/S0097539796300921}
Bernstein, E. and Vazirani, U.
\newblock Quantum complexity theory.
\newblock \emph{SIAM Journal on Computing}, 26\penalty0 (5):\penalty0
  1411--1473, 1997.
\newblock \doi{10.1137/S0097539796300921}.
\newblock URL \url{https://epubs.siam.org/doi/10.1137/S0097539796300921}.

\bibitem[Beylkin(1987)]{1165108}
Beylkin, G.
\newblock Discrete radon transform.
\newblock \emph{IEEE Transactions on Acoustics, Speech, and Signal Processing},
  35\penalty0 (2):\penalty0 162--172, 1987.
\newblock \doi{10.1109/TASSP.1987.1165108}.
\newblock URL \url{https://ieeexplore.ieee.org/document/1165108}.

\bibitem[Biamonte et~al.(2017)Biamonte, Wittek, Pancotti, Rebentrost, Wiebe,
  and Lloyd]{biamonte2017quantum}
Biamonte, J., Wittek, P., Pancotti, N., Rebentrost, P., Wiebe, N., and Lloyd,
  S.
\newblock Quantum machine learning.
\newblock \emph{Nature}, 549\penalty0 (7671):\penalty0 195, 2017.
\newblock URL \url{https://www.nature.com/articles/nature23474}.

\bibitem[Boag et~al.(2000)Boag, Bresler, and Michielssen]{862638}
Boag, A., Bresler, Y., and Michielssen, E.
\newblock A multilevel domain decomposition algorithm for fast o(n/sup 2/logn)
  reprojection of tomographic images.
\newblock \emph{IEEE Transactions on Image Processing}, 9\penalty0
  (9):\penalty0 1573--1582, 2000.
\newblock \doi{10.1109/83.862638}.
\newblock URL \url{https://ieeexplore.ieee.org/document/862638}.

\bibitem[Brady(1998)]{doi:10.1137/S0097539793256673}
Brady, M.~L.
\newblock A fast discrete approximation algorithm for the radon transform.
\newblock \emph{SIAM Journal on Computing}, 27\penalty0 (1):\penalty0 107--119,
  1998.
\newblock \doi{10.1137/S0097539793256673}.
\newblock URL \url{https://epubs.siam.org/doi/10.1137/S0097539793256673}.

\bibitem[Brandt et~al.(2000)Brandt, Mann, Brodski, and
  Galun]{doi:10.1137/S003613999732425X}
Brandt, A., Mann, J., Brodski, M., and Galun, M.
\newblock A fast and accurate multilevel inversion of the radon transform.
\newblock \emph{SIAM Journal on Applied Mathematics}, 60\penalty0 (2):\penalty0
  437--462, 2000.
\newblock \doi{10.1137/S003613999732425X}.
\newblock URL \url{https://epubs.siam.org/doi/10.1137/S003613999732425X}.

\bibitem[Candes(1998)]{candes}
Candes, E.~J.
\newblock \emph{Ridgelets : theory and applications}.
\newblock PhD thesis, Stanford University, 1998.
\newblock URL \url{https://searchworks.stanford.edu/view/9949708}.

\bibitem[Carre \& Andres(2004)Carre and Andres]{CARRE20042165}
Carre, P. and Andres, E.
\newblock Discrete analytical ridgelet transform.
\newblock \emph{Signal Processing}, 84\penalty0 (11):\penalty0 2165--2173,
  2004.
\newblock ISSN 0165-1684.
\newblock \doi{https://doi.org/10.1016/j.sigpro.2004.07.009}.
\newblock URL
  \url{https://www.sciencedirect.com/science/article/pii/S0165168404001689}.
\newblock Special Section Signal Processing in Communications.

\bibitem[Chakraborty et~al.(2019)Chakraborty, Gily{\'e}n, and
  Jeffery]{chakraborty_et_al:LIPIcs:2019:10609}
Chakraborty, S., Gily{\'e}n, A., and Jeffery, S.
\newblock {The Power of Block-Encoded Matrix Powers: Improved Regression
  Techniques via Faster Hamiltonian Simulation}.
\newblock In Baier, C., Chatzigiannakis, I., Flocchini, P., and Leonardi, S.
  (eds.), \emph{46th International Colloquium on Automata, Languages, and
  Programming (ICALP 2019)}, volume 132 of \emph{Leibniz International
  Proceedings in Informatics (LIPIcs)}, pp.\  33:1--33:14, Dagstuhl, Germany,
  2019. Schloss Dagstuhl--Leibniz-Zentrum fuer Informatik.
\newblock ISBN 978-3-95977-109-2.
\newblock \doi{10.4230/LIPIcs.ICALP.2019.33}.
\newblock URL \url{http://drops.dagstuhl.de/opus/volltexte/2019/10609}.

\bibitem[Chen et~al.(2022{\natexlab{a}})Chen, Cotler, Huang, and Li]{9719827}
Chen, S., Cotler, J., Huang, H., and Li, J.
\newblock Exponential separations between learning with and without quantum
  memory.
\newblock In \emph{2021 IEEE 62nd Annual Symposium on Foundations of Computer
  Science (FOCS)}, pp.\  574--585, Los Alamitos, CA, USA, feb
  2022{\natexlab{a}}. IEEE Computer Society.
\newblock \doi{10.1109/FOCS52979.2021.00063}.
\newblock URL
  \url{https://doi.ieeecomputersociety.org/10.1109/FOCS52979.2021.00063}.

\bibitem[Chen et~al.(2022{\natexlab{b}})Chen, Zhang, and
  Wang]{chen2022peekaboo}
Chen, X., Zhang, J., and Wang, Z.
\newblock Peek-a-boo: What (more) is disguised in a randomly weighted neural
  network, and how to find it efficiently.
\newblock In \emph{International Conference on Learning Representations},
  2022{\natexlab{b}}.
\newblock URL \url{https://openreview.net/forum?id=moHCzz6D5H3}.

\bibitem[Childs et~al.(2017)Childs, Kothari, and Somma]{doi:10.1137/16M1087072}
Childs, A.~M., Kothari, R., and Somma, R.~D.
\newblock Quantum algorithm for systems of linear equations with exponentially
  improved dependence on precision.
\newblock \emph{SIAM Journal on Computing}, 46\penalty0 (6):\penalty0
  1920--1950, 2017.
\newblock \doi{10.1137/16M1087072}.
\newblock URL \url{https://epubs.siam.org/doi/10.1137/16M1087072}.

\bibitem[Ciliberto et~al.(2018)Ciliberto, Herbster, Ialongo, Pontil, Rocchetto,
  Severini, and Wossnig]{doi:10.1098/rspa.2017.0551}
Ciliberto, C., Herbster, M., Ialongo, A.~D., Pontil, M., Rocchetto, A.,
  Severini, S., and Wossnig, L.
\newblock Quantum machine learning: a classical perspective.
\newblock \emph{Proceedings of the Royal Society A: Mathematical, Physical and
  Engineering Sciences}, 474\penalty0 (2209):\penalty0 20170551, 2018.
\newblock \doi{10.1098/rspa.2017.0551}.
\newblock URL
  \url{https://royalsocietypublishing.org/doi/abs/10.1098/rspa.2017.0551}.

\bibitem[Cleve \& Watrous(2000)Cleve and Watrous]{892140}
Cleve, R. and Watrous, J.
\newblock Fast parallel circuits for the quantum fourier transform.
\newblock In \emph{Proceedings 41st Annual Symposium on Foundations of Computer
  Science}, pp.\  526--536, 2000.
\newblock \doi{10.1109/SFCS.2000.892140}.
\newblock URL \url{https://ieeexplore.ieee.org/document/892140}.

\bibitem[Coppersmith(1994)]{https://doi.org/10.48550/arxiv.quant-ph/0201067}
Coppersmith, D.
\newblock An approximate fourier transform useful in quantum factoring.
\newblock \emph{IBM Research Report}, pp.\  RC--19642, 1994.
\newblock URL \url{https://dominoweb.draco.res.ibm.com/reports/8472.ps.gz}.

\bibitem[de~Wolf(2019)]{arXiv:1907.09415}
de~Wolf, R.
\newblock Quantum computing: Lecture notes, 2019.
\newblock URL \url{https://arxiv.org/abs/1907.09415}.

\bibitem[Do \& Vetterli(2003)Do and Vetterli]{1187351}
Do, M. and Vetterli, M.
\newblock The finite ridgelet transform for image representation.
\newblock \emph{IEEE Transactions on Image Processing}, 12\penalty0
  (1):\penalty0 16--28, 2003.
\newblock \doi{10.1109/TIP.2002.806252}.
\newblock URL \url{https://ieeexplore.ieee.org/document/1187351}.

\bibitem[Donoho(1993)]{donoho1993unconditional}
Donoho, D.~L.
\newblock Unconditional bases are optimal bases for data compression and for
  statistical estimation.
\newblock \emph{Applied and Computational Harmonic Analysis}, 1\penalty0
  (1):\penalty0 100--115, 1993.
\newblock ISSN 1063-5203.
\newblock \doi{https://doi.org/10.1006/acha.1993.1008}.
\newblock URL
  \url{https://www.sciencedirect.com/science/article/pii/S1063520383710080}.

\bibitem[Fadili \& Starck(2012)Fadili and Starck]{Fadili2012}
Fadili, J. and Starck, J.-L.
\newblock \emph{Curvelets and Ridgelets}, pp.\  754--773.
\newblock Springer New York, New York, NY, 2012.
\newblock ISBN 978-1-4614-1800-9.
\newblock \doi{10.1007/978-1-4614-1800-9_48}.
\newblock URL
  \url{https://link.springer.com/referenceworkentry/10.1007/978-1-4614-1800-9_48}.

\bibitem[Fijany \& Williams(1998)Fijany and Williams]{10.5555/645812.670803}
Fijany, A. and Williams, C.~P.
\newblock Quantum wavelet transforms: Fast algorithms and complete circuits.
\newblock In \emph{Selected Papers from the First NASA International Conference
  on Quantum Computing and Quantum Communications}, QCQC '98, pp.\  10–33,
  Berlin, Heidelberg, 1998. Springer-Verlag.
\newblock ISBN 354065514X.
\newblock URL \url{https://dl.acm.org/doi/abs/10.5555/645812.670803}.

\bibitem[Frankle \& Carbin(2019)Frankle and Carbin]{frankle2018the}
Frankle, J. and Carbin, M.
\newblock The lottery ticket hypothesis: Finding sparse, trainable neural
  networks.
\newblock In \emph{International Conference on Learning Representations}, 2019.
\newblock URL \url{https://openreview.net/forum?id=rJl-b3RcF7}.

\bibitem[Frankle et~al.(2021)Frankle, Dziugaite, Roy, and
  Carbin]{frankle2021pruning}
Frankle, J., Dziugaite, G.~K., Roy, D., and Carbin, M.
\newblock Pruning neural networks at initialization: Why are we missing the
  mark?
\newblock In \emph{International Conference on Learning Representations}, 2021.
\newblock URL \url{https://openreview.net/forum?id=Ig-VyQc-MLK}.

\bibitem[Gily\'{e}n et~al.(2019)Gily\'{e}n, Su, Low, and
  Wiebe]{10.1145/3313276.3316366}
Gily\'{e}n, A., Su, Y., Low, G.~H., and Wiebe, N.
\newblock Quantum singular value transformation and beyond: Exponential
  improvements for quantum matrix arithmetics.
\newblock In \emph{Proceedings of the 51st Annual ACM SIGACT Symposium on
  Theory of Computing}, STOC 2019, pp.\  193–204, New York, NY, USA, 2019.
  Association for Computing Machinery.
\newblock ISBN 9781450367059.
\newblock \doi{10.1145/3313276.3316366}.
\newblock URL \url{https://dl.acm.org/doi/10.1145/3313276.3316366}.

\bibitem[Giovannetti et~al.(2008{\natexlab{a}})Giovannetti, Lloyd, and
  Maccone]{PhysRevA.78.052310}
Giovannetti, V., Lloyd, S., and Maccone, L.
\newblock Architectures for a quantum random access memory.
\newblock \emph{Phys. Rev. A}, 78:\penalty0 052310, Nov 2008{\natexlab{a}}.
\newblock \doi{10.1103/PhysRevA.78.052310}.
\newblock URL \url{https://link.aps.org/doi/10.1103/PhysRevA.78.052310}.

\bibitem[Giovannetti et~al.(2008{\natexlab{b}})Giovannetti, Lloyd, and
  Maccone]{PhysRevLett.100.160501}
Giovannetti, V., Lloyd, S., and Maccone, L.
\newblock Quantum random access memory.
\newblock \emph{Phys. Rev. Lett.}, 100:\penalty0 160501, Apr
  2008{\natexlab{b}}.
\newblock \doi{10.1103/PhysRevLett.100.160501}.
\newblock URL \url{https://link.aps.org/doi/10.1103/PhysRevLett.100.160501}.

\bibitem[Goodfellow et~al.(2016)Goodfellow, Bengio, and
  Courville]{Goodfellow-et-al-2016}
Goodfellow, I., Bengio, Y., and Courville, A.
\newblock \emph{Deep Learning}.
\newblock MIT Press, 2016.
\newblock URL \url{http://www.deeplearningbook.org}.

\bibitem[Grover \& Rudolph(2002)Grover and
  Rudolph]{https://doi.org/10.48550/arxiv.quant-ph/0208112}
Grover, L. and Rudolph, T.
\newblock Creating superpositions that correspond to efficiently integrable
  probability distributions, 2002.
\newblock URL \url{https://arxiv.org/abs/quant-ph/0208112}.

\bibitem[Götz \& Druckmüller(1996)Götz and Druckmüller]{GOTZ1996711}
Götz, W. and Druckmüller, H.
\newblock A fast digital radon transform—an efficient means for evaluating
  the hough transform.
\newblock \emph{Pattern Recognition}, 29\penalty0 (4):\penalty0 711--718, 1996.
\newblock ISSN 0031-3203.
\newblock \doi{https://doi.org/10.1016/0031-3203(96)00015-5}.
\newblock URL
  \url{https://www.sciencedirect.com/science/article/pii/0031320396000155}.

\bibitem[Hales \& Hallgren(2000)Hales and Hallgren]{892139}
Hales, L. and Hallgren, S.
\newblock An improved quantum fourier transform algorithm and applications.
\newblock In \emph{Proceedings 41st Annual Symposium on Foundations of Computer
  Science}, pp.\  515--525, 2000.
\newblock \doi{10.1109/SFCS.2000.892139}.
\newblock URL \url{https://ieeexplore.ieee.org/document/892139}.

\bibitem[Hann et~al.(2021)Hann, Lee, Girvin, and Jiang]{PRXQuantum.2.020311}
Hann, C.~T., Lee, G., Girvin, S., and Jiang, L.
\newblock Resilience of quantum random access memory to generic noise.
\newblock \emph{PRX Quantum}, 2:\penalty0 020311, Apr 2021.
\newblock \doi{10.1103/PRXQuantum.2.020311}.
\newblock URL \url{https://link.aps.org/doi/10.1103/PRXQuantum.2.020311}.

\bibitem[Harvey et~al.(2019)Harvey, Liaw, Plan, and
  Randhawa]{pmlr-v99-harvey19a}
Harvey, N. J.~A., Liaw, C., Plan, Y., and Randhawa, S.
\newblock Tight analyses for non-smooth stochastic gradient descent.
\newblock In Beygelzimer, A. and Hsu, D. (eds.), \emph{Proceedings of the
  Thirty-Second Conference on Learning Theory}, volume~99 of \emph{Proceedings
  of Machine Learning Research}, pp.\  1579--1613. PMLR, 25--28 Jun 2019.
\newblock URL \url{https://proceedings.mlr.press/v99/harvey19a.html}.

\bibitem[Hayakawa \& Suzuki(2020)Hayakawa and Suzuki]{hayakawa2020minimax}
Hayakawa, S. and Suzuki, T.
\newblock On the minimax optimality and superiority of deep neural network
  learning over sparse parameter spaces.
\newblock \emph{Neural Networks}, 123:\penalty0 343--361, 2020.
\newblock ISSN 0893-6080.
\newblock \doi{https://doi.org/10.1016/j.neunet.2019.12.014}.
\newblock URL
  \url{https://www.sciencedirect.com/science/article/pii/S089360801930406X}.

\bibitem[Helbert et~al.(2006)Helbert, Carre, and Andres]{4011958}
Helbert, D., Carre, P., and Andres, E.
\newblock 3-d discrete analytical ridgelet transform.
\newblock \emph{IEEE Transactions on Image Processing}, 15\penalty0
  (12):\penalty0 3701--3714, 2006.
\newblock \doi{10.1109/TIP.2006.881936}.
\newblock URL \url{https://ieeexplore.ieee.org/abstract/document/4011958}.

\bibitem[Hoyer(1997)]{https://doi.org/10.48550/arxiv.quant-ph/9702028}
Hoyer, P.
\newblock Efficient quantum transforms, 1997.
\newblock URL \url{https://arxiv.org/abs/quant-ph/9702028}.

\bibitem[Huang et~al.(2021)Huang, Kueng, and Preskill]{PhysRevLett.126.190505}
Huang, H.-Y., Kueng, R., and Preskill, J.
\newblock Information-theoretic bounds on quantum advantage in machine
  learning.
\newblock \emph{Phys. Rev. Lett.}, 126:\penalty0 190505, May 2021.
\newblock \doi{10.1103/PhysRevLett.126.190505}.
\newblock URL \url{https://link.aps.org/doi/10.1103/PhysRevLett.126.190505}.

\bibitem[Huang et~al.(2022)Huang, Broughton, Cotler, Chen, Li, Mohseni, Neven,
  Babbush, Kueng, Preskill, and McClean]{doi:10.1126/science.abn7293}
Huang, H.-Y., Broughton, M., Cotler, J., Chen, S., Li, J., Mohseni, M., Neven,
  H., Babbush, R., Kueng, R., Preskill, J., and McClean, J.~R.
\newblock Quantum advantage in learning from experiments.
\newblock \emph{Science}, 376\penalty0 (6598):\penalty0 1182--1186, 2022.
\newblock \doi{10.1126/science.abn7293}.
\newblock URL \url{https://www.science.org/doi/abs/10.1126/science.abn7293}.

\bibitem[Kelley \& Madisetti(1993)Kelley and Madisetti]{236530}
Kelley, B. and Madisetti, V.
\newblock The fast discrete radon transform. i. theory.
\newblock \emph{IEEE Transactions on Image Processing}, 2\penalty0
  (3):\penalty0 382--400, 1993.
\newblock \doi{10.1109/83.236530}.
\newblock URL \url{https://ieeexplore.ieee.org/document/236530}.

\bibitem[Kerenidis \& Prakash(2017)Kerenidis and
  Prakash]{kerenidis_et_al:LIPIcs:2017:8154}
Kerenidis, I. and Prakash, A.
\newblock {Quantum Recommendation Systems}.
\newblock In Papadimitriou, C.~H. (ed.), \emph{8th Innovations in Theoretical
  Computer Science Conference (ITCS 2017)}, volume~67 of \emph{Leibniz
  International Proceedings in Informatics (LIPIcs)}, pp.\  49:1--49:21,
  Dagstuhl, Germany, 2017. Schloss Dagstuhl--Leibniz-Zentrum fuer Informatik.
\newblock ISBN 978-3-95977-029-3.
\newblock \doi{10.4230/LIPIcs.ITCS.2017.49}.
\newblock URL \url{http://drops.dagstuhl.de/opus/volltexte/2017/8154}.

\bibitem[Kitaev(1995)]{https://doi.org/10.48550/arxiv.quant-ph/9511026}
Kitaev, A.~Y.
\newblock Quantum measurements and the abelian stabilizer problem, 1995.
\newblock URL \url{https://arxiv.org/abs/quant-ph/9511026}.

\bibitem[Labunets et~al.(2001{\natexlab{a}})Labunets, Labunets-Rundblad, and
  Astola]{938692}
Labunets, V., Labunets-Rundblad, E., and Astola, J.
\newblock Fast classical and quantum fractional haar wavelet transforms.
\newblock In \emph{ISPA 2001. Proceedings of the 2nd International Symposium on
  Image and Signal Processing and Analysis. In conjunction with 23rd
  International Conference on Information Technology Interfaces (IEEE Cat.},
  pp.\  564--569, 2001{\natexlab{a}}.
\newblock \doi{10.1109/ISPA.2001.938692}.
\newblock URL \url{https://ieeexplore.ieee.org/abstract/document/938692}.

\bibitem[Labunets et~al.(2001{\natexlab{b}})Labunets, Labunets-Rundblad,
  Egiazarian, and Astola]{938691}
Labunets, V., Labunets-Rundblad, E., Egiazarian, K., and Astola, J.
\newblock Fast classical and quantum fractional walsh transforms.
\newblock In \emph{ISPA 2001. Proceedings of the 2nd International Symposium on
  Image and Signal Processing and Analysis. In conjunction with 23rd
  International Conference on Information Technology Interfaces (IEEE Cat.},
  pp.\  558--563, 2001{\natexlab{b}}.
\newblock \doi{10.1109/ISPA.2001.938691}.
\newblock URL \url{https://ieeexplore.ieee.org/abstract/document/938691}.

\bibitem[Lee et~al.(2019)Lee, Ajanthan, and Torr]{lee2018snip}
Lee, N., Ajanthan, T., and Torr, P.
\newblock {SNIP}: {SINGLE}-{SHOT} {NETWORK} {PRUNING} {BASED} {ON} {CONNECTION}
  {SENSITIVITY}.
\newblock In \emph{International Conference on Learning Representations}, 2019.
\newblock URL \url{https://openreview.net/forum?id=B1VZqjAcYX}.

\bibitem[Li et~al.(2019)Li, Fan, ying Xia, and Song]{LI2019113}
Li, H.-S., Fan, P., ying Xia, H., and Song, S.
\newblock Quantum multi-level wavelet transforms.
\newblock \emph{Information Sciences}, 504:\penalty0 113--135, 2019.
\newblock ISSN 0020-0255.
\newblock \doi{https://doi.org/10.1016/j.ins.2019.07.057}.
\newblock URL
  \url{https://www.sciencedirect.com/science/article/pii/S0020025519306632}.

\bibitem[Li et~al.(2022)Li, Fan, Peng, Song, and Long]{9337176}
Li, H.-S., Fan, P., Peng, H., Song, S., and Long, G.-L.
\newblock Multilevel 2-d quantum wavelet transforms.
\newblock \emph{IEEE Transactions on Cybernetics}, 52\penalty0 (8):\penalty0
  8467--8480, 2022.
\newblock \doi{10.1109/TCYB.2021.3049509}.
\newblock URL \url{https://ieeexplore.ieee.org/abstract/document/9337176}.

\bibitem[Liu et~al.(2021)Liu, Arunachalam, and Temme]{Yunchao2020}
Liu, Y., Arunachalam, S., and Temme, K.
\newblock A rigorous and robust quantum speed-up in supervised machine
  learning.
\newblock \emph{Nature Physics}, 17:\penalty0 1013, 2021.
\newblock URL \url{https://www.nature.com/articles/s41567-021-01287-z}.

\bibitem[Liu(2009)]{10.1145/1536414.1536469}
Liu, Y.-K.
\newblock Quantum algorithms using the curvelet transform.
\newblock In \emph{Proceedings of the Forty-First Annual ACM Symposium on
  Theory of Computing}, STOC '09, pp.\  391–400, New York, NY, USA, 2009.
  Association for Computing Machinery.
\newblock ISBN 9781605585062.
\newblock \doi{10.1145/1536414.1536469}.
\newblock URL \url{https://dl.acm.org/doi/10.1145/1536414.1536469}.

\bibitem[L{\"{o}}fberg(2004)]{Lofberg2004}
L{\"{o}}fberg, J.
\newblock Yalmip : A toolbox for modeling and optimization in matlab.
\newblock In \emph{In Proceedings of the CACSD Conference}, Taipei, Taiwan,
  2004.
\newblock URL \url{https://yalmip.github.io/}.

\bibitem[Ma et~al.(2022)Ma, Li, and Zhao]{9648027}
Ma, G., Li, H., and Zhao, J.
\newblock Quantum radon transforms and their applications.
\newblock \emph{IEEE Transactions on Quantum Engineering}, 3:\penalty0 1--16,
  2022.
\newblock \doi{10.1109/TQE.2021.3134648}.
\newblock URL \url{https://ieeexplore.ieee.org/abstract/document/9648027}.

\bibitem[Malach et~al.(2020)Malach, Yehudai, Shalev-Schwartz, and
  Shamir]{pmlr-v119-malach20a}
Malach, E., Yehudai, G., Shalev-Schwartz, S., and Shamir, O.
\newblock Proving the lottery ticket hypothesis: Pruning is all you need.
\newblock In III, H.~D. and Singh, A. (eds.), \emph{Proceedings of the 37th
  International Conference on Machine Learning}, volume 119 of
  \emph{Proceedings of Machine Learning Research}, pp.\  6682--6691. PMLR,
  13--18 Jul 2020.
\newblock URL \url{https://proceedings.mlr.press/v119/malach20a.html}.

\bibitem[Matteo et~al.(2020)Matteo, Gheorghiu, and Mosca]{8962352}
Matteo, O.~D., Gheorghiu, V., and Mosca, M.
\newblock Fault-tolerant resource estimation of quantum random-access memories.
\newblock \emph{IEEE Transactions on Quantum Engineering}, 1:\penalty0 1--13,
  2020.
\newblock \doi{10.1109/TQE.2020.2965803}.
\newblock URL \url{https://ieeexplore.ieee.org/document/8962352}.

\bibitem[Mosca \& Zalka(2004)Mosca and Zalka]{doi:10.1142/S0219749904000109}
Mosca, M. and Zalka, C.
\newblock Exact quantum fourier transforms and discrete logarithm algorithms.
\newblock \emph{International Journal of Quantum Information}, 02\penalty0
  (01):\penalty0 91--100, 2004.
\newblock \doi{10.1142/S0219749904000109}.
\newblock URL
  \url{https://www.worldscientific.com/doi/abs/10.1142/S0219749904000109}.

\bibitem[Murata(1996)]{MURATA1996947}
Murata, N.
\newblock An integral representation of functions using three-layered networks
  and their approximation bounds.
\newblock \emph{Neural Networks}, 9\penalty0 (6):\penalty0 947--956, 1996.
\newblock ISSN 0893-6080.
\newblock \doi{https://doi.org/10.1016/0893-6080(96)00000-7}.
\newblock URL
  \url{https://www.sciencedirect.com/science/article/pii/0893608096000007}.

\bibitem[Nielsen \& Chuang(2011)Nielsen and Chuang]{N4}
Nielsen, M.~A. and Chuang, I.~L.
\newblock \emph{Quantum Computation and Quantum Information: 10th Anniversary
  Edition}.
\newblock Cambridge University Press, 10th edition, 2011.
\newblock ISBN 9781107002173.
\newblock URL
  \url{https://www.cambridge.org/highereducation/books/quantum-computation-and-quantum-information/01E10196D0A682A6AEFFEA52D53BE9AE}.

\bibitem[Orseau et~al.(2020)Orseau, Hutter, and
  Rivasplata]{NEURIPS2020_1e949147}
Orseau, L., Hutter, M., and Rivasplata, O.
\newblock Logarithmic pruning is all you need.
\newblock In Larochelle, H., Ranzato, M., Hadsell, R., Balcan, M., and Lin, H.
  (eds.), \emph{Advances in Neural Information Processing Systems}, volume~33,
  pp.\  2925--2934. Curran Associates, Inc., 2020.
\newblock URL
  \url{https://proceedings.neurips.cc/paper/2020/file/1e9491470749d5b0e361ce4f0b24d037-Paper.pdf}.

\bibitem[Pensia et~al.(2020)Pensia, Rajput, Nagle, Vishwakarma, and
  Papailiopoulos]{NEURIPS2020_1b742ae2}
Pensia, A., Rajput, S., Nagle, A., Vishwakarma, H., and Papailiopoulos, D.
\newblock Optimal lottery tickets via subset sum: Logarithmic
  over-parameterization is sufficient.
\newblock In Larochelle, H., Ranzato, M., Hadsell, R., Balcan, M., and Lin, H.
  (eds.), \emph{Advances in Neural Information Processing Systems}, volume~33,
  pp.\  2599--2610. Curran Associates, Inc., 2020.
\newblock URL
  \url{https://proceedings.neurips.cc/paper/2020/file/1b742ae215adf18b75449c6e272fd92d-Paper.pdf}.

\bibitem[Press(2006)]{doi:10.1073/pnas.0609228103}
Press, W.~H.
\newblock Discrete radon transform has an exact, fast inverse and generalizes
  to operations other than sums along lines.
\newblock \emph{Proceedings of the National Academy of Sciences}, 103\penalty0
  (51):\penalty0 19249--19254, 2006.
\newblock \doi{10.1073/pnas.0609228103}.
\newblock URL \url{https://www.pnas.org/doi/abs/10.1073/pnas.0609228103}.

\bibitem[Rahimi \& Recht(2008)Rahimi and Recht]{R2}
Rahimi, A. and Recht, B.
\newblock Random features for large-scale kernel machines.
\newblock In Platt, J.~C., Koller, D., Singer, Y., and Roweis, S.~T. (eds.),
  \emph{Advances in Neural Information Processing Systems 20}, pp.\
  1177--1184. Curran Associates, Inc., 2008.
\newblock URL
  \url{http://papers.nips.cc/paper/3182-random-features-for-large-scale-kernel-machines.pdf}.

\bibitem[Rahimi \& Recht(2009)Rahimi and Recht]{R3}
Rahimi, A. and Recht, B.
\newblock Weighted sums of random kitchen sinks: Replacing minimization with
  randomization in learning.
\newblock In Koller, D., Schuurmans, D., Bengio, Y., and Bottou, L. (eds.),
  \emph{Advances in Neural Information Processing Systems 21}, pp.\
  1313--1320. Curran Associates, Inc., 2009.
\newblock URL
  \url{http://papers.nips.cc/paper/3495-weighted-sums-of-random-kitchen-sinks-replacing-minimization-with-randomization-in-learning.pdf}.

\bibitem[Ramanujan et~al.(2020)Ramanujan, Wortsman, Kembhavi, Farhadi, and
  Rastegari]{Ramanujan_2020_CVPR}
Ramanujan, V., Wortsman, M., Kembhavi, A., Farhadi, A., and Rastegari, M.
\newblock What’s hidden in a randomly weighted neural network?
\newblock In \emph{2020 IEEE/CVF Conference on Computer Vision and Pattern
  Recognition (CVPR)}, pp.\  11890--11899, Los Alamitos, CA, USA, jun 2020.
  IEEE Computer Society.
\newblock \doi{10.1109/CVPR42600.2020.01191}.
\newblock URL
  \url{https://doi.ieeecomputersociety.org/10.1109/CVPR42600.2020.01191}.

\bibitem[Rubin(1998)]{Rubin1998}
Rubin, B.
\newblock The calderón reproducing formula, windowed x-ray transforms, and
  radon transforms in l...-spaces.
\newblock \emph{The journal of Fourier analysis and applications
  [[Elektronische Ressource]]}, 4\penalty0 (2):\penalty0 175--198, 1998.
\newblock URL \url{https://link.springer.com/article/10.1007/BF02475988}.

\bibitem[Schuld \& Petruccione(2021)Schuld and Petruccione]{schuld2021machine}
Schuld, M. and Petruccione, F.
\newblock \emph{Machine learning with quantum computers}.
\newblock Springer, 2021.
\newblock URL \url{https://link.springer.com/book/10.1007/978-3-030-83098-4}.

\bibitem[Shor(1997)]{doi:10.1137/S0097539795293172}
Shor, P.~W.
\newblock Polynomial-time algorithms for prime factorization and discrete
  logarithms on a quantum computer.
\newblock \emph{SIAM Journal on Computing}, 26\penalty0 (5):\penalty0
  1484--1509, 1997.
\newblock \doi{10.1137/S0097539795293172}.
\newblock URL \url{https://epubs.siam.org/doi/10.1137/S0097539795293172}.

\bibitem[Simon(1994)]{365701}
Simon, D.
\newblock On the power of quantum computation.
\newblock In \emph{Proceedings 35th Annual Symposium on Foundations of Computer
  Science}, pp.\  116--123, 1994.
\newblock \doi{10.1109/SFCS.1994.365701}.
\newblock URL \url{https://ieeexplore.ieee.org/document/365701}.

\bibitem[Sonoda \& Murata(2014)Sonoda and Murata]{10.1007/978-3-319-11179-7_68}
Sonoda, S. and Murata, N.
\newblock Sampling hidden parameters from oracle distribution.
\newblock In Wermter, S., Weber, C., Duch, W., Honkela, T.,
  Koprinkova-Hristova, P., Magg, S., Palm, G., and Villa, A. E.~P. (eds.),
  \emph{Artificial Neural Networks and Machine Learning -- ICANN 2014}, pp.\
  539--546, Cham, 2014. Springer International Publishing.
\newblock ISBN 978-3-319-11179-7.
\newblock URL
  \url{https://link.springer.com/chapter/10.1007/978-3-319-11179-7_68}.

\bibitem[Sonoda \& Murata(2017)Sonoda and Murata]{SONODA2017233}
Sonoda, S. and Murata, N.
\newblock Neural network with unbounded activation functions is universal
  approximator.
\newblock \emph{Applied and Computational Harmonic Analysis}, 43\penalty0
  (2):\penalty0 233--268, 2017.
\newblock ISSN 1063-5203.
\newblock \doi{https://doi.org/10.1016/j.acha.2015.12.005}.
\newblock URL
  \url{https://www.sciencedirect.com/science/article/pii/S1063520315001748}.

\bibitem[Sonoda et~al.(2021)Sonoda, Ishikawa, and Ikeda]{pmlr-v130-sonoda21a}
Sonoda, S., Ishikawa, I., and Ikeda, M.
\newblock Ridge regression with over-parametrized two-layer networks converge
  to ridgelet spectrum.
\newblock In Banerjee, A. and Fukumizu, K. (eds.), \emph{Proceedings of The
  24th International Conference on Artificial Intelligence and Statistics},
  volume 130 of \emph{Proceedings of Machine Learning Research}, pp.\
  2674--2682. PMLR, 13--15 Apr 2021.
\newblock URL \url{https://proceedings.mlr.press/v130/sonoda21a.html}.

\bibitem[Sonoda et~al.(2022{\natexlab{a}})Sonoda, Ishikawa, and
  Ikeda]{pmlr-v162-sonoda22a}
Sonoda, S., Ishikawa, I., and Ikeda, M.
\newblock Fully-connected network on noncompact symmetric space and ridgelet
  transform based on helgason-{F}ourier analysis.
\newblock In Chaudhuri, K., Jegelka, S., Song, L., Szepesvari, C., Niu, G., and
  Sabato, S. (eds.), \emph{Proceedings of the 39th International Conference on
  Machine Learning}, volume 162 of \emph{Proceedings of Machine Learning
  Research}, pp.\  20405--20422. PMLR, 17--23 Jul 2022{\natexlab{a}}.
\newblock URL \url{https://proceedings.mlr.press/v162/sonoda22a.html}.

\bibitem[Sonoda et~al.(2022{\natexlab{b}})Sonoda, Ishikawa, and
  Ikeda]{sonoda2022universality}
Sonoda, S., Ishikawa, I., and Ikeda, M.
\newblock Universality of group convolutional neural networks based on ridgelet
  analysis on groups.
\newblock In Oh, A.~H., Agarwal, A., Belgrave, D., and Cho, K. (eds.),
  \emph{Advances in Neural Information Processing Systems}, 2022{\natexlab{b}}.
\newblock URL \url{https://openreview.net/forum?id=ebCk2FNI1za}.

\bibitem[Starck et~al.(2010)Starck, Murtagh, and
  Fadili]{starck_murtagh_fadili_2010}
Starck, J.-L., Murtagh, F., and Fadili, J.~M.
\newblock \emph{The Ridgelet and Curvelet Transforms}, pp.\  89--118.
\newblock Cambridge University Press, 2010.
\newblock \doi{10.1017/CBO9780511730344.006}.
\newblock URL
  \url{https://www.cambridge.org/core/books/abs/sparse-image-and-signal-processing/ridgelet-and-curvelet-transforms/D84CAD0CEB84E2940B43B4ED298C0C62}.

\bibitem[Sweke et~al.(2021)Sweke, Seifert, Hangleiter, and
  Eisert]{Sweke2021quantumversus}
Sweke, R., Seifert, J.-P., Hangleiter, D., and Eisert, J.
\newblock On the {Q}uantum versus {C}lassical {L}earnability of {D}iscrete
  {D}istributions.
\newblock \emph{{Quantum}}, 5:\penalty0 417, March 2021.
\newblock ISSN 2521-327X.
\newblock \doi{10.22331/q-2021-03-23-417}.
\newblock URL \url{https://quantum-journal.org/papers/q-2021-03-23-417/}.

\bibitem[Taha(2016)]{Taha2016}
Taha, S. M.~R.
\newblock \emph{Wavelets and Multiwavelets Implementation Using Quantum
  Computing}, pp.\  153--170.
\newblock Springer International Publishing, Cham, 2016.
\newblock ISBN 978-3-319-23479-3.
\newblock \doi{10.1007/978-3-319-23479-3_7}.
\newblock URL
  \url{https://link.springer.com/chapter/10.1007/978-3-319-23479-3_7}.

\bibitem[Tanaka et~al.(2020)Tanaka, Kunin, Yamins, and
  Ganguli]{NEURIPS2020_46a4378f}
Tanaka, H., Kunin, D., Yamins, D.~L., and Ganguli, S.
\newblock Pruning neural networks without any data by iteratively conserving
  synaptic flow.
\newblock In Larochelle, H., Ranzato, M., Hadsell, R., Balcan, M., and Lin, H.
  (eds.), \emph{Advances in Neural Information Processing Systems}, volume~33,
  pp.\  6377--6389. Curran Associates, Inc., 2020.
\newblock URL
  \url{https://proceedings.neurips.cc/paper/2020/file/46a4378f835dc8040c8057beb6a2da52-Paper.pdf}.

\bibitem[Tang(2019)]{10.1145/3313276.3316310}
Tang, E.
\newblock A quantum-inspired classical algorithm for recommendation systems.
\newblock In \emph{Proceedings of the 51st Annual ACM SIGACT Symposium on
  Theory of Computing}, STOC 2019, pp.\  217–228, New York, NY, USA, 2019.
  Association for Computing Machinery.
\newblock ISBN 9781450367059.
\newblock \doi{10.1145/3313276.3316310}.
\newblock URL \url{https://dl.acm.org/doi/10.1145/3313276.3316310}.

\bibitem[Tseng \& Hwang(2004)Tseng and Hwang]{1328767}
Tseng, C.-C. and Hwang, T.-M.
\newblock Quantum circuit design of 8 /spl times/ 8 discrete hartley transform.
\newblock In \emph{2004 IEEE International Symposium on Circuits and Systems
  (ISCAS)}, volume~3, pp.\  III--397, 2004.
\newblock \doi{10.1109/ISCAS.2004.1328767}.
\newblock URL \url{https://ieeexplore.ieee.org/abstract/document/1328767}.

\bibitem[Wang et~al.(2020{\natexlab{a}})Wang, Zhang, and
  Grosse]{Wang2020Picking}
Wang, C., Zhang, G., and Grosse, R.
\newblock Picking winning tickets before training by preserving gradient flow.
\newblock In \emph{International Conference on Learning Representations},
  2020{\natexlab{a}}.
\newblock URL \url{https://openreview.net/forum?id=SkgsACVKPH}.

\bibitem[Wang et~al.(2020{\natexlab{b}})Wang, Zhang, Xie, Zhou, Su, Zhang, and
  Hu]{Wang_Zhang_Xie_Zhou_Su_Zhang_Hu_2020}
Wang, Y., Zhang, X., Xie, L., Zhou, J., Su, H., Zhang, B., and Hu, X.
\newblock Pruning from scratch.
\newblock \emph{Proceedings of the AAAI Conference on Artificial Intelligence},
  34\penalty0 (07):\penalty0 12273--12280, Apr. 2020{\natexlab{b}}.
\newblock \doi{10.1609/aaai.v34i07.6910}.
\newblock URL \url{https://ojs.aaai.org/index.php/AAAI/article/view/6910}.

\bibitem[Wortsman et~al.(2020)Wortsman, Ramanujan, Liu, Kembhavi, Rastegari,
  Yosinski, and Farhadi]{NEURIPS2020_ad1f8bb9}
Wortsman, M., Ramanujan, V., Liu, R., Kembhavi, A., Rastegari, M., Yosinski,
  J., and Farhadi, A.
\newblock Supermasks in superposition.
\newblock In Larochelle, H., Ranzato, M., Hadsell, R., Balcan, M., and Lin, H.
  (eds.), \emph{Advances in Neural Information Processing Systems}, volume~33,
  pp.\  15173--15184. Curran Associates, Inc., 2020.
\newblock URL
  \url{https://proceedings.neurips.cc/paper/2020/file/ad1f8bb9b51f023cdc80cf94bb615aa9-Paper.pdf}.

\bibitem[Yamakawa \& Zhandry(2022)Yamakawa and Zhandry]{9996892}
Yamakawa, T. and Zhandry, M.
\newblock Verifiable quantum advantage without structure.
\newblock In \emph{2022 IEEE 63rd Annual Symposium on Foundations of Computer
  Science (FOCS)}, pp.\  69--74, Los Alamitos, CA, USA, nov 2022. IEEE Computer
  Society.
\newblock \doi{10.1109/FOCS54457.2022.00014}.
\newblock URL
  \url{https://doi.ieeecomputersociety.org/10.1109/FOCS54457.2022.00014}.

\bibitem[Yamasaki \& Sonoda(2021)Yamasaki and
  Sonoda]{https://doi.org/10.48550/arxiv.2106.09028}
Yamasaki, H. and Sonoda, S.
\newblock Exponential error convergence in data classification with optimized
  random features: Acceleration by quantum machine learning, 2021.
\newblock URL \url{https://arxiv.org/abs/2106.09028}.

\bibitem[Yamasaki et~al.(2020)Yamasaki, Subramanian, Sonoda, and
  Koashi]{NEURIPS2020_9ddb9dd5}
Yamasaki, H., Subramanian, S., Sonoda, S., and Koashi, M.
\newblock Learning with optimized random features: Exponential speedup by
  quantum machine learning without sparsity and low-rank assumptions.
\newblock In Larochelle, H., Ranzato, M., Hadsell, R., Balcan, M., and Lin, H.
  (eds.), \emph{Advances in Neural Information Processing Systems}, volume~33,
  pp.\  13674--13687. Curran Associates, Inc., 2020.
\newblock URL
  \url{https://proceedings.neurips.cc/paper/2020/file/9ddb9dd5d8aee9a76bf217a2a3c54833-Paper.pdf}.

\bibitem[Zhou et~al.(2019)Zhou, Lan, Liu, and Yosinski]{NEURIPS2019_1113d7a7}
Zhou, H., Lan, J., Liu, R., and Yosinski, J.
\newblock Deconstructing lottery tickets: Zeros, signs, and the supermask.
\newblock In Wallach, H., Larochelle, H., Beygelzimer, A., d\textquotesingle
  Alch\'{e}-Buc, F., Fox, E., and Garnett, R. (eds.), \emph{Advances in Neural
  Information Processing Systems}, volume~32. Curran Associates, Inc., 2019.
\newblock URL
  \url{https://proceedings.neurips.cc/paper/2019/file/1113d7a76ffceca1bb350bfe145467c6-Paper.pdf}.

\end{thebibliography}

\newpage

\appendix
\onecolumn

\section*{Appendices}

Appendices are organized as follows.
In Appendix~\ref{sec:A}, we provide a proof of Theorem~\ref{thm:fourier_slice_theorem} in Sec.~\ref{sec:discrete_ridgelet_transform}.
In Appendix~\ref{sec:B}, we provide a proof of Theorem~\ref{thm:ridgelet} in Sec.~\ref{sec:exact_reconstruction}.
In Appendix~\ref{sec:C}, we summarize basic notions of quantum computation required for the analysis and then provide a proof of Theorem~\ref{thm:runtime_qrt} in Sec.~\ref{sec:quantum_ridgelet_transform}.
In Appendix~\ref{sec:D}, we clarify the explicit implementation and the runtime of the input model for our quantum algorithm for Theorem~\ref{thm:quantum_sampling} in Sec.~\ref{sec:winning} and then provide the proof of Theorem~\ref{thm:quantum_sampling}.
In Appendix~\ref{sec:E}, we provide the proof of Theorem~\ref{thm:winning} in Sec.~\ref{sec:finding}.
In Appendix~\ref{sec:F}, we describe the detail of the parameters chosen for the numerical experiment in Sec.~\ref{sec:advantage}.

\section{\label{sec:A}Proof on Fourier Slice Theorem for Discrete Ridgelet Transform}

We provide a proof of Theorem~\ref{thm:fourier_slice_theorem} in Sec.~\ref{sec:discrete_ridgelet_transform}.

\begin{proof}[Proof of Theorem~\ref{thm:fourier_slice_theorem}]
    It holds that
    \begin{align}
            &\mathcal{R}[f](\va,b)\\
            &=P^{-\frac{D}{2}}\sum_{\vx}f(\vx)r((\va^\top \vx- b)\bmod P)\\
            &=P^{-\frac{D}{2}}\sum_{\vx}f(\vx)\overline{r((\va^\top \vx- b)\bmod
            P)}\\
            &=P^{-\frac{D}{2}}\sum_{\vx}f(\vx)\,\overline{P^{-\frac{1}{2}}\sum_{v}\mathcal{F}_1[r](v)\mathrm{e}^{\frac{2\pi\mathrm{i}v(\va^\top
            \vx- b)}{P}}}\\
            &=P^{-\frac{1}{2}}\sum_{v\in\mathbb{Z}_P}\big(P^{-\frac{D}{2}}\sum_{\vx}f(\vx)\mathrm{e}^{\frac{-2\pi\mathrm{i}v\va^\top \vx}{P}}\big)\,\overline{\mathcal{F}_1[r](v)}\mathrm{e}^{\frac{2\pi\mathrm{i}v b}{P}}\\
            &=P^{-\frac{1}{2}}\sum_{v\in\mathbb{Z}_P}\mathcal{F}_D[f](v \va\bmod P)\;\overline{\mathcal{F}_1[r](v)}\;\mathrm{e}^{\frac{2\pi\mathrm{i}v b}{P}}.
  \end{align}
  Therefore, the discrete Fourier transform for $b$ leads to
  \begin{equation}
      \mathcal{F}_1[\mathcal{R}[f](\va,\cdot)]=\mathcal{F}_D[f](v \va\bmod P)\;\overline{\mathcal{F}_1[r](v)},
  \end{equation}
  which yields the conclusion.
\end{proof}

\section{\label{sec:B}Proof on Exact Representation of Function by Discretized Neural Network}

We provide a proof of Theorem~\ref{thm:ridgelet} in Sec.~\ref{sec:exact_reconstruction}.

\begin{proof}[Proof of Theorem~\ref{thm:ridgelet}]
In this proof, we write $w=\mathcal{R}[f]$ for simplicity of notation.
    By~\eqref{eq:S}, we have
    \begin{align}
        &\mathcal{S}[w](\vx)=P^{-\frac{D}{2}}\sum_{\va}[w(\va,\cdot)\ast
        g(\cdot)](\va^\top \vx)\\
        \label{eq:2}
        &=P^{-\frac{D}{2}}\sum_{\va}\sum_{v}\mathcal{F}_1[w(\va,\cdot)](v)\;\mathcal{F}_1[g](v)\;\mathrm{e}^{\frac{2\pi\mathrm{i}v\va^\top \vx}{P}},
    \end{align}
    where we use
    \begin{equation}
        f\ast g(y)\coloneqq\sum_{b}f(b)g(y-b)=\sum_{v}\mathcal{F}_1[f](v)\mathcal{F}_1[g](v)\mathrm{e}^{\frac{2\pi\mathrm{i}vy}{P}}.
    \end{equation}
    Then, applying Theorem~\ref{thm:fourier_slice_theorem} to $w=\mathcal{R}[f]$, we have
    \begin{align}
        \label{eq:3}
        &\text{\eqref{eq:2}}=P^{-\frac{D}{2}}\sum_{\va}\sum_{v}\mathcal{F}_D[f](v \va\bmod P)\;\overline{\mathcal{F}_1[r](v)}\;\mathcal{F}_1[g](v)\;\mathrm{e}^{\frac{2\pi\mathrm{i}v\va^\top \vx}{G}}.
    \end{align}
    To evaluate the right-hand side, we use the fact that $P$ is a prime. In this case, for any $v\neq 0$, the set $\sZ_P=\{0,1,\ldots,P-1\}$ is identical to $\{0,v,\ldots,(P-1)v\}$.
    Thus, recalling $\mathcal{F}_1[g](0)=\sum_b g(b)=0$ as in~\eqref{eq:g_not_zero} and letting $\va^\prime\in\sZ_P^D$ denote the unique vector satisfying $\va^\prime=v \va\bmod P\in\sZ_P^D$, we obtain
    \begin{align}
    &\text{\eqref{eq:3}}=P^{\nicefrac{-D}{2}}\sum_{v\in\sZ_P}\sum_{\va^\prime\in\sZ_P^D}\mathcal{F}_D[f](\va^\prime)\;\overline{\mathcal{F}_1[r](v)}\;\mathcal{F}_1[g](v)\;\mathrm{e}^{\frac{2\pi\mathrm{i}\va^{\prime\top}
    \vx}{P}}\\
    &=\left(\sum_v\mathcal{F}_1[g](v)\overline{\mathcal{F}_1[r](v)}\right)\times\left(P^{\nicefrac{-D}{2}}\sum_{\va^\prime}\mathcal{F}_D[f](\va^\prime)\mathrm{e}^{\frac{2\pi\mathrm{i}\va^{\prime\top}\vx}{P}}\right)\\
    &=C_{g,r}f(\vx).
    \end{align}
    Therefore, it holds that
    \begin{equation}
        f(\vx)=\frac{1}{C_{g,r}}\mathcal{S}[w](\vx),
    \end{equation}
    which yields the conclusion.
\end{proof}

\section{\label{sec:C}Proof on Runtime of Quantum Ridgelet Transform}

In this section, we summarize basic notions of quantum computation; for more detial, see the textbooks and the lecture notes by, e.g.,~\citet{N4,arXiv:1907.09415}.
Then, we provide a proof of Theorem~\ref{thm:runtime_qrt} in Sec.~\ref{sec:quantum_ridgelet_transform}.

Analogously to a bit $\{0,1\}$ in classical computation,
the unit of quantum computation is a quantum bit (qubit), mathematically represented by $\sC^2$, i.e., a $2$-dimensional complex Hilbert space.
A fixed orthonormal basis of a qubit $\sC^2$ is denoted by $\left\{\Ket{0}\coloneqq\left(\begin{smallmatrix}1\\0\end{smallmatrix}\right),\Ket{1}\coloneqq\left(\begin{smallmatrix}0\\1\end{smallmatrix}\right)\right\}$.
Similar to a bit taking a state $b\in\{0,1\}$, a qubit takes a quantum state $\Ket{\psi}=\alpha_0\Ket{0}+\alpha_1\Ket{1}=\left(\begin{smallmatrix}\alpha_0\\\alpha_1\end{smallmatrix}\right)\in\mathbb{C}^2$.
While a register of $m$ bits takes values in ${\left\{0,1\right\}}^m$, a quantum register of $m$ qubits is represented by the tensor-product space ${\left(\sC^2\right)}^{\otimes m}\cong\mathbb{C}^{2^m}$, i.e., a $2^m$-dimensional Hilbert space.
We may use $=$ rather than $\cong$ to represent isomorphism for brevity.
We let $\mathcal{H}$ denote a finite-dimensional Hilbert space representing a quantum register; that is, an $m$-qubit register is $\mathcal{H}=\mathbb{C}^{2^m}$.
A fixed orthonormal basis $\{\Ket{x}: x\in\{0,\ldots,2^m-1\}\}$ labeled by $m$-bit strings, or the corresponding integers, is called the \textit{standard basis} of $\mathcal{H}$.
A state of $\mathcal{H}$ can be denoted by $\Ket{\psi} = \sum_{x=0}^{2^m-1}\alpha_x\Ket{x}\in\mathcal{H}$.
Note that any quantum state $\Ket{\psi}$ requires an $L^2$ normalization condition $\left\|\Ket{\psi}\right\|_2=1$, and for any $\theta\in\mathbb{R}$, $\Ket{\psi}$ is identified with $\mathrm{e}^{\mathrm{i}\theta}\Ket{\psi}$.

In the bra-ket notation,
the conjugate transpose of the column vector $\Ket{\psi}$ is a row vector denoted by $\Bra{\psi}$, where $\Bra{\psi}$ and $\Ket{\psi}$ may be called a bra and a ket, respectively. The inner product of $\Ket{\psi}$ and $\Ket{\phi}$ is denoted by $\Braket{\psi|\phi}$, while their outer product $\Ket{\psi}\Bra{\phi}$ is a matrix.
The conjugate transpose of a matrix $\mA$ is denoted by $\mA^\dag$, and the transpose of $\mA$ with respect to the standard basis is denoted by $\mA^\top$.

A measurement of a quantum state $\Ket{\psi}$ is a sampling process that returns a randomly chosen bit string from the quantum state.
An $m$-qubit state $\Ket{\psi}=\sum_{x=0}^{2^m-1}\alpha_x\Ket{x}$ is said to be in a superposition of the basis states $\Ket{x}$s.
A measurement of $\Ket{\psi}$ in the standard basis $\{\Ket{x}\}$ provides a random $m$-bit integer $x\in\{0,\ldots,2^m-1\}$ as outcome, with probability $p(x)={|\alpha_x|}^2$.
After the measurement, the state changes from $\Ket{\psi}$ to $\Ket{x}$ corresponding to the obtained outcome $x$, and loses the randomness in $\Ket{\psi}$; that is, to iterate the same sampling as this measurement, we need to prepare $\Ket{\psi}$ repeatedly for each iteration.
For two registers $\mathcal{H}^A\otimes\mathcal{H}^B$ and their state $\Ket{\phi}^{AB}=\sum_{x,x}\alpha_{x,x^\prime}\Ket{x}^A\otimes\Ket{x^\prime}^B\in\mathcal{H}^A\otimes\mathcal{H}^B$, a measurement of the register $\mathcal{H}^B$ for $\Ket{\phi}^{AB}$ in the standard basis $\{\Ket{x^\prime}^B\}$ of $\mathcal{H}^B$ yields an outcome $x^\prime$ with probability $p(x^\prime)=\sum_{x}p(x,x^\prime)$, where $p(x,x^\prime)={|\alpha_{x,x^\prime}|}^2$.
The superscripts of a state or an operator represent which register the state or the operator belongs to, while we may omit the superscripts if it is clear from the context.

A quantum algorithm starts by initializing $m$ qubits in a fixed state $\Ket{0}^{\otimes m}$, which we may write as $\Ket{0}$ if $m$ is clear from the context.
Then, we apply a $2^m$-dimensional unitary operator $\mU$ to $\Ket{0}^{\otimes m}$, to prepare a state $\mU\Ket{0}^{\otimes m}$.
Finally, a measurement of $\mU\Ket{0}^{\otimes m}$ is performed to sample an $m$-bit integer from a probability distribution given by $\mU\Ket{0}^{\otimes m}$.
Analogously to classical logic-gate circuits, $\mU$ is represented by a quantum circuit composed of sequential applications of unitaries acting at most two qubits at a time.
Each of these unitaries is called an elementary quantum gate.
The runtime of a quantum algorithm represented by a quantum circuit is determined by the number of applications of elementary quantum gates in the circuit.

Using these notions, the proof of Theorem~\ref{thm:runtime_qrt} in Sec.~\ref{sec:quantum_ridgelet_transform} is shown as follows.

\begin{proof}[Proof of Theorem~\ref{thm:runtime_qrt}]
    The runtime of Algorithm~\ref{alg:qrt} is domianated by Step~2 as shown in the following.

    Step~1 is performed within runtime $O(\polylog(P))$ due to our assumption that $\ket{r}=\sum_{b}r(b)\ket{b}$ can be prepared in time $O(\polylog(P))$.

    Step~2 is dominated by the runtime of $\mF_P^{\otimes D}$, which is $O(D\times\polylog(P))$ since the runtime of $\mF_P$ is $O(\polylog(P))$ as shown in~\eqref{eq:QFT}. The inverse $\mF_P^\dag$ has the same runtime as $\mF_P$ since $\mF_P^\dag$ is implemented by applying the inverse of each gate in the quantum circuit for $\mF_P$ in the reverse order.
    A techinical remark is that, for simplicity of presentation, we write our statement and the proof based on the exact implementation of $\mF_P$ by~\citet{doi:10.1142/S0219749904000109} to avoid writing the polylogarithmic error factors that may arise in approximate implementations of $\mF_P$ such as those by~\citet{https://doi.org/10.48550/arxiv.quant-ph/9511026,892139}.
    Even if one uses these approximate implementations of $\mF_P$, our theorem follows from the same argument with the polylogarithmic error factors multiplied.
    We also note that some of the other implementations of QFT such as those by~\citet{https://doi.org/10.48550/arxiv.quant-ph/0201067,892140} are targeted at $P=2^n$ for $n=1,2,\ldots$, and our algorithm does not use these implementations since $P$ is a prime number in our setting.

    Step~3 is performed within runtime $O(\polylog(P))$ by arithmetics on $O(\log(P))$ qubits. This is implemented by writing the classical computation of the arithmetic in a reversible way as a classical circuit and replacing each Toffoli gate in the classical circuit with the quantum Toffoli gate to obtain the quantum circuit for the arithmetics part.

    Step~4 has the runtime of $O(\polylog(P))$ for $\mF_P^\dag$ as discussed in Step~2 as well.

    Consequently, the overall runtime is dominated by that of Step~2, i.e., $O(D\times\polylog(P))$.
\end{proof}

\section{\label{sec:D}Proof on Runtime of Quantum Algorithm for Sampling from Optimized Probability Distribution}

In this section, we clarify the explicit implementation and the runtime of the input model for our algorithm of Theorem~\ref{thm:quantum_sampling} in Sec.~\ref{sec:winning}.
Then, We provide the proof of Theorem~\ref{thm:quantum_sampling}.
Our algorithm for Theorem~\ref{thm:quantum_sampling} is shown in Algorithm~\ref{alg:sampling}.
See also Appendix~\ref{sec:C} for the notations on quantum computation.

\begin{algorithm}[t]
  \caption{\label{alg:sampling}Quantum algorithm for sampling from optimized probability distribution for
winning ticket of neural networks.}
  \begin{algorithmic}[1]
      \REQUIRE{$\lambda,\Delta>0$, assumptions in Sec.~\ref{sec:setting}, the input model described in Sec.~\ref{sec:winning}.}
      \ENSURE{Parameters $(\va,b)$ sampled from the optimized probability distribution $p_{\lambda,\Delta}^\ast(\va,b)$ in~\eqref{eq:optimized_distribution}.}
      \STATE{Prepare a quantum state $\ket{\psi_\mathrm{in}}\propto\sum_{\vx}\hat{p}_\mathrm{data}(\vx)f(\vx)\ket{\vx}$.}
      \STATE{Apply QRT to obtain a state $\mR\ket{\psi_\mathrm{in}}$.}
      \STATE{Apply ${(\mR\hat{\mP}_\mathrm{data}\mR^\top+\lambda\mI)}^{-1}$ by quantum singular value transformation (QSVT) to obtain
      \begin{align}
          &\frac{1}{\sqrt{\gamma}}\sum_{\va,b}P^{-\frac{D}{2}}w_\lambda^\ast(\va,b)\ket{\va,b}\propto{(\mR\hat{\mP}_\mathrm{data}\mR^\top+\lambda\mI)}^{-1}\mR\ket{\psi_\mathrm{in}}.
      \end{align}}
      \STATE{Apply ${(\mW_{\lambda}+\frac{\Delta}{\gamma}\mI)}^{-\frac{1}{2}}$ by QSVT to obtain
      \begin{align}
          &\ket{p_{\lambda,\Delta}^\ast} \propto{\Big(\mW_{\lambda}+\frac{\Delta}{\gamma}\mI\Big)}^{-\frac{1}{2}}{(\mR\hat{\mP}_\mathrm{data}\mR^\top+\lambda\mI)}^{-1}\mR\ket{\psi_\mathrm{in}}.
      \end{align}
      }
  \STATE{Perform a measurement in the standard basis $\{\ket{\va,b}\}$ to sample $(\va,b)$ as the outcome according to the probability distribution $p_{\lambda,\Delta}^\ast(\va,b)$.}
    \STATE{\textbf{Return } $(\va,b)$.}
  \end{algorithmic}
\end{algorithm}

As an input model,
Algorithm~\ref{alg:sampling} uses preparation of quantum states
\begin{align}
    \ket{\hat{p}_\mathrm{data}}&\coloneqq\sum_{\vx}\sqrt{\hat{p}_\mathrm{data}(\vx)}\ket{\vx},\\
    \ket{\psi_\mathrm{in}}&\coloneqq\frac{\sum_{\vx}\hat{p}_\mathrm{data}(\vx)f(\vx)\ket{\vx}}{\sqrt{\sum_{\vx}{|\hat{p}_\mathrm{data}(\vx)f(\vx)|}^2}},
\end{align}
where $\hat{p}_\mathrm{data}$ is the empirical distribution of $\vx_1,\ldots,\vx_M$  for the $M$ input-output pairs of examples $(\vx_1,f(\vx_1)),\ldots,(\vx_M,f(\vx_M))\in\sZ_P^D\times\sR$.
We will explain the implementation of the input model by quantum circuit, so as to show that the runtime of the input model in terms of the circuit depth can be bounded by $O(D\polylog(M))$ including the runtime of quantum random access memory (QRAM)~\citep{PhysRevA.78.052310,PhysRevLett.100.160501} used for the implementation.
In summary, upon collecting the $M$ examples, we will construct an $O(M)$-size sparse data structure shown by~\citet{kerenidis_et_al:LIPIcs:2017:8154}, and the state preparation can be efficiently conducted with the algorithm by~\citet{https://doi.org/10.48550/arxiv.quant-ph/0208112} using this data structure combined with QRAM\@; in addition, it is known that QRAM used for this input model is implementable explicitly as a parallelized quantum circuit of a $O(\polylog(M))$ depth using at most $O(M)$ qubits as shown, e.g., by~\citet{8962352,PRXQuantum.2.020311}.
In the following, we explain the detail of these facts to avoid any potential confusion about feasibly of the input model for Algorithm~\ref{alg:sampling} in our setting.

We explain how to construct the sparse data structure shown by~\citet{kerenidis_et_al:LIPIcs:2017:8154} in our setting.
In the following, we describe the data structure for $\ket{\hat{p}_\mathrm{data}}$ in detail, and then clarify the difference between $\ket{\hat{p}_\mathrm{data}}$ and $\ket{\psi_\mathrm{in}}$.
For the preparation of $\ket{\hat{p}_\mathrm{data}}$, along with collecting $(\vx_m,f(\vx_m))$ one by one for each $m\in\{1,\ldots,M\}$, we are to perform a preprocessing to count the number of examples and store the empirical distribution in a sparse data structure proposed by~\citet{kerenidis_et_al:LIPIcs:2017:8154}.
To describe this sparse data structure, we first explain the underlying dense binary tree (which we introduce for explanation but never store in the memory), and then clarify a sparse version of the binary tree to be stored in the memory.
Each leaf of the underlying dense binary tree represents a set $\{\vx\}$ for each $\vx\in\sZ_P^D$; i.e., the underlying dense binary tree has $O(P^D)$ leaves.
Each parent in this dense binary tree represents the sum of the two sets represented by its two children; thus, the root of the tree represents $\sZ_P^D$.
With this definition, each node in the dense binary tree aims at storing the number of examples in the set represented by the node; e.g., a leaf representing $\{\vx\}$ stores the cardinality of $\{m\in\{1,\ldots,M\}:\vx_m=\vx\}$.
As for the sparse version of this binary tree,
each leaf for $\{\vx\}$ counts and stores the number of examples satisfying $\vx_m=\vx$ in the same way, but the sparse data structure does not store leaves with zero example.
Each parent in the sparse data structure stores the sum of the counts for its children in the tree in the same way, but the sparse data structure does not store branches with zero example.
As a whole, the root should store the number of all the examples, i.e., $M$.
In this sparse data structure, the number of nodes at each depth of the tree is at most $M$, and the height of the tree is $O(D\log(P))$.
To construct this sparse data structure for $\ket{\hat{p}_\mathrm{data}}$, we initialize the counts in all the nodes as $0$, and for each $\vx_m$ of the $M$ examples, we increment the count in the leaf for $\{\vx_m\}$ and those in the corresponding ascendants of this leaf, where the increment for each example can be performed within poly-logarithmic time $O(\polylog(M))$ as shown by~\citet{kerenidis_et_al:LIPIcs:2017:8154}.
Therefore, the runtime of this preprocessing for all $M$ examples is $\widetilde{O}(M)$, where $\widetilde{O}$ ignores the poly-logarithmic factors;
that is, this runtime has the same scaling as just collecting the $M$ examples up to the poly-logarithmic factors.
The sparse data structure for $\ket{\psi_\mathrm{in}}$ is based on the same underlying dense binary tree, but a leaf for each $\{\vx\}$ stores $\hat{p}_\mathrm{data}(\vx)f(\vx)$ instead of $\hat{p}_\mathrm{data}(\vx)$, and each parent store the sum of those at its children in the same way; to construct this data structure, instead of incrementing the counts stored in the nodes by $+1$ for each $\vx_m$, we add $f(\vx_m)$ to the number stored in a leaf $\{\vx_m\}$ and update the numbers stored in the ascendants of this leaf correspondingly.

We use this sparse data structure with QRAM to prepare the states $\ket{\hat{p}_\mathrm{data}}$ and $\ket{\psi_\mathrm{in}}$.
In particular, the preparation of these states is achieved by a parallelized quantum circuit of depth $O(D\log(P))$ to run a quantum algorithm of~\citet{https://doi.org/10.48550/arxiv.quant-ph/0208112}, where each $O(1)$-depth part of the circuit processes each depth of the $O(D\log(P))$-height tree, and the QRAM is queried once for each of these parts, in total $O(D\log(P))$ times~\citep{kerenidis_et_al:LIPIcs:2017:8154}.
Each of the queries to the QRAM is implementable within runtime $O(\polylog(M))$.
In particular, the QRAM is an architecture for using classical data in a quantum algorithm without destroying superposition, defined as (1) in the work by~\citet{PhysRevLett.100.160501}.
In our setting, the sparse data structure has at most $M$ nodes as leaves, and the number of nodes at each depth of the tree has at most $1,2,4,8,\ldots,M$ nodes, respectively.
Each query to the QRAM performs a quantum circuit depending on one of the collections of these $1,2,4,8,\ldots,M$ nodes at each depth of the tree.
The number of nodes to be used in each query is smaller than $M$.
Then, the QRAM for these $O(M)$ nodes is implementable by a $O(\polylog(M))$-depth parallelized quantum circuit on $O(M)$ qubits per query, e.g., by a circuit given in Fig.~10 of the work by~\citet{PRXQuantum.2.020311}, which is queried for each depth of the tree for our sparse data structure.
Importantly, the QRAM never measures and reads out classical bit values stored in the $O(M)$ nodes at each depth, but just performs an $O(\polylog(M))$-depth sequence of unitary gates in parallel to maintain quantum superposition.
As a whole, given the data structure and QRAM, the overall runtime per preparation of $\ket{\hat{p}_\mathrm{data}}$ and $\ket{\psi_\mathrm{in}}$ is bounded by
\begin{equation}
\label{eq:runtime_preparation}
  O(D\log(P)\times\polylog(M))=\widetilde{O}(D).
\end{equation}
To summarize,  this runtime is achievable because  of the following facts;
\begin{itemize}
    \item the $M$ examples are stored in the sparse data structure representing the binary tree of height $O(D\log(P))$;
    \item each depth of the $O(D\log(P))$-height tree is processed by a constant-depth quantum circuit with one query to the QRAM\@;
    \item  a single query to the QRAM for processing the $O(M)$ nodes at each depth of the tree is implemented by the $O(\polylog(M))$-depth quantum circuit.
\end{itemize}

With this input model, we provide the proof of Theorem~\ref{thm:quantum_sampling} on the runtime of Algorithm~\ref{alg:sampling}.

\begin{proof}
    In Algorithm~\ref{alg:sampling}, by steps~1,~2, and~3, we prepare a quantum state proportional to
    \begin{equation}
        \sum_{\va,b}P^{-\frac{D}{2}}w_\lambda^\ast(\va,b)\ket{\va,b}={(\mR\hat{\mP}_\mathrm{data}\mR^\top+\lambda\mI)}^{-1}\mR\sum_{\vx}\hat{p}_\mathrm{data}(\vx)f(\vx)\ket{\vx},
    \end{equation}
    which is the analytical formula for the solution of ridge regression in~\eqref{eq:minimization}.
    Then, by step~4 to apply ${(\mW_{\lambda}+\frac{\Delta}{\gamma}\mI)}^{-\frac{1}{2}}$ to this state, we obtain a quantum state
    \begin{equation}
        \ket{p_{\lambda,\Delta}^\ast}=\frac{\sum_{\va,b}P^{-\frac{D}{2}}w_\lambda^\ast(\va,b)\ket{\va,b}}{\sqrt{\sum_{\va,b}|P^{-\frac{D}{2}}w_\lambda^\ast(\va,b)|^2+\Delta}}.
    \end{equation}
    The runtime of Algorithm~\ref{alg:sampling} is dominated by step~4 requiring
    \begin{equation}
        \widetilde{O}\left(\frac{D}{\lambda}\frac{1}{\nicefrac{\Delta}{\gamma}}\right),
    \end{equation}
    as we will show step by step in the following.

    Step~1 for preparing the state $\ket{\psi_\mathrm{in}}$ is performed within runtime
    \begin{equation}
        \widetilde{O}(D),
    \end{equation}
    due to~\eqref{eq:runtime_preparation}.

    Step~2 for performing QRT is performed by Algorithm~\ref{alg:qrt} within runtime shown in Theorem~\ref{thm:runtime_qrt}, i.e.,
    \begin{equation}
        \widetilde{O}(D).
    \end{equation}

    Step~3 for applying ${(\mR\hat{\mP}_\mathrm{data}\mR^\top+\lambda\mI)}^{-1}$ is performed within runtime $\widetilde{O}(\nicefrac{D}{\lambda})$, as shown in the following.
    To explain Step~3, we follow a unified framework for describing a general class of quantum algorithms based on quantum singular value transform (QSVT) developed by~\citet{10.1145/3313276.3316366}; in particular, we construct a block encoding of $\mR\hat{\mP}_\mathrm{data}\mR^\top+\lambda\mI$ and implement ${(\mR\hat{\mP}_\mathrm{data}\mR^\top+\lambda\mI)}^{-1}$ by applying QSVT to this block encoding.
    Our construction of the block encoding of $\mR\hat{\mP}_\mathrm{data}\mR^\top+\lambda\mI$ is based on linear combination of block-encoded matrices $\mR\hat{\mP}_\mathrm{data}\mR^\top$ and $\mI$~\citep{10.1145/3313276.3316366}.
    The block encoding of $\mI$ is a trivial unitary operator $\mI$, and thus we here show the block encoding of $\mR\hat{\mP}_\mathrm{data}\mR^\top$.
    For $\mR\hat{\mP}_\mathrm{data}\mR^\top$, we use a block encoding of a density operator $\rho=\mR\hat{\mP}_\mathrm{data}\mR^\top$; in particular, if we have a quantum circuit $U$ for preparing a purification of $\rho$ from $\ket{0}$, then we can construct a block encoding of $\rho$ using $U$ and $U^\dag$ a constant number of times~\citep{10.1145/3313276.3316366}.
    We here prepare the purification of $\rho$ as follows.
    First, we prepare
    \begin{equation}
    \label{eq:state_qrt_data}
        \sum_{\vx}\sqrt{\hat{p}_\mathrm{data}(\vx)}\ket{\vx}
    \end{equation}
    within runtime
    \begin{equation}
        \widetilde{O}(D),
    \end{equation}
    due to~\eqref{eq:runtime_preparation}.
    Then, we add the same number of auxiliary qubits initialized in $\ket{0}$ as those for the state~\eqref{eq:state_qrt_data}, and apply \textsc{CNOT} gates on the auxiliary qubits controlled by this state to obtain
    \begin{equation}
    \label{eq:55}
        \sum_{\vx}\sqrt{\hat{p}_\mathrm{data}(\vx)}\ket{\vx}\otimes\ket{\vx},
    \end{equation}
    which requires runtime
    \begin{equation}
        \widetilde{O}(D)
    \end{equation}
    since the number of qubits is $\widetilde{O}(D)$.
    Then, we apply QRT to the first register to obtain
    \begin{equation}
        \Big(\mR\sum_{\vx}\sqrt{\hat{p}_\mathrm{data}(\vx)}\ket{\vx}\Big)\otimes\ket{\vx}
    \end{equation}
    within runtime shown in Theorem~\ref{thm:runtime_qrt}, i.e.,
    \begin{equation}
        \widetilde{O}(D).
    \end{equation}
    By tracing out the auxiliary qubits, the reduced state of the obtained state is $\rho=\mR\hat{\mP}_\mathrm{data}\mR^\top$, and the runtime of this block encoding is $\widetilde{O}(D)$.
    As a result, using the block encodings of $\mR\hat{\mP}_\mathrm{data}\mR^\top$ and $\mI$ a constant number of times, the procedure for linear combination of the block encoded matrices $\mR\hat{\mP}_\mathrm{data}\mR^\top$ and $\mI$ yields the block encoding of $\mR\hat{\mP}_\mathrm{data}\mR^\top+\lambda\mI$, which has the runtime
    \begin{equation}
        \label{eq:runtime_block_encoding_3}
        \widetilde{O}(D)
    \end{equation}
    dominated by that of $\mR\hat{\mP}_\mathrm{data}\mR^\top$.

    To apply ${(\mR\hat{\mP}_\mathrm{data}\mR^\top+\lambda\mI)}^{-1}$ in Step~3 to the state prepared by Step~2,
    we use QSVT of the block encoding of $\mR\hat{\mP}_\mathrm{data}\mR^\top+\lambda\mI$.
    As shown by~\citet{10.1145/3313276.3316366}, given any state $\ket{\psi}$ and a circuit $\mU$ for a block encoding of $\mA$, we can prepare a state proportional to $\mA^{-1}\ket{\psi}$ by querying $\mU$ and $\mU^\dag$ in total $\widetilde{O}(\kappa)$ times using the technique of variable-time amplitude amplification~\cite{ambainis:LIPIcs:2012:3426,doi:10.1137/16M1087072,chakraborty_et_al:LIPIcs:2019:10609,10.1145/3313276.3316366}, where $\kappa$ is the condition number of $\mA$, and the runtime includes that for amplitude amplification.
    In our case, the condition number of $\mA=\mR\hat{\mP}_\mathrm{data}\mR^\top+\lambda\mI$ is
    \begin{equation}
        \kappa\leqq\frac{1+\lambda}{\lambda}=O\Big(\frac{1}{\lambda}\Big), 
    \end{equation}
    since the largest eigenvalue $\|\mA\|_\infty$ of this $\mA$ is upper bounded by
    \begin{equation}
        \|\mA\|_\infty\leqq\|\mR\|_\infty\|\hat{\mP}_\mathrm{data}\|_\infty\|\mR^\top\|_\infty+\|\lambda\mI\|_\infty\leqq 1+\lambda,
    \end{equation}
    and the smallest eigenvalue is lower bounded by $\lambda$ due to the term $\lambda\mI$.
    The runtime of each circuit $\mU$ for this block encoding is $\widetilde{O}(D)$ as shown in~\eqref{eq:runtime_block_encoding_3}.
    As a whole, the runtime by the end of Step~3 is dominated by
    \begin{equation}
    \label{eq:runtime_step3}
        \widetilde{O}(\kappa)\times\widetilde{O}(D)=\widetilde{O}\Big(\frac{D}{\lambda}\Big).
    \end{equation}

    Step~4 for applying ${(\mW_{\lambda}+\frac{\Delta}{\gamma}\mI)}^{-\frac{1}{2}}$ is performed within runtime $\widetilde{O}\Big(\frac{D}{\lambda}\times \frac{1}{\nicefrac{\Delta}{\gamma}}\Big)$, as shown in the following.
    As in the above explanation, using the framework of QSVT, we here construct a block encoding of $\mW_{\lambda}+\frac{\Delta}{\gamma}\mI$ and implement ${\Big(\mW_{\lambda}+\frac{\Delta}{\gamma}\mI\Big)}^{-\frac{1}{2}}$ by applying QSVT to this block encoding.
    Our construction of the block encoding of $\mW_{\lambda}+\frac{\Delta}{\gamma}\mI$ is based on linear combination of block-encoded matrices $\mW_{\lambda}$ and $\mI$.
    The block encoding of $\mI$ is a trivial unitary operator, and noticing that $\mW_{\lambda}$ is a density operator by definition (i.e.,  $\mW_{\lambda}\geqq 0$ and  $\Tr[\mW_{\lambda}]=1$), we show the block encoding of $\mW_{\lambda}$ using the block encoding of the density operator similar to the above.
    In particular, by the end of Step~3, we have constructed a circuit for preparing
    \begin{align}
        &\frac{1}{\sqrt{\gamma}}\sum_{\va,b}P^{-\frac{D}{2}}w^\ast(\va,b)\ket{\va,b}\propto{(\mR\hat{\mP}_\mathrm{data}\mR^\top+\lambda\mI)}^{-1}\mR\sum_{\vx}\sqrt{\hat{p}_\mathrm{data}(\vx)}\ket{\vx},
    \end{align}
    which runs within time $\widetilde{O}(\frac{D}{\lambda})$ due to~\eqref{eq:runtime_step3}.
    Then, in the same way as~\eqref{eq:55}, we add the same number of auxiliary qubits initialized in $\ket{0}$ as those for this state, and apply \textsc{CNOT} gates on the auxiliary qubits controlled by the above state to obtain
    \begin{equation}
        \frac{1}{\sqrt{\gamma}}\sum_{\va,b}P^{-\frac{D}{2}}w^\ast(\va,b)\ket{\va,b}\otimes\ket{\va,b},
    \end{equation}
    which requires $\widetilde{O}(D)$ runtime since the number of qubits is $\widetilde{O}(D)$.
    By tracing out the auxiliary qubits, the reduced state becomes
    \begin{equation}
        \mW_{\lambda}=\sum_{\va,b}\frac{|P^{-\frac{D}{2}}w^\ast(\va,b)|^2}{\gamma}\ket{\va,b}\bra{\va,b},
    \end{equation}
    and the runtime of this block encoding is $\widetilde{O}(\nicefrac{D}{\lambda})$ dominated by the state-preparation circuit by the end of Step~3\@.
    Then, using the block encodings of $\mW_{\lambda}$ and $\mI$ a constant number of times, the linear combination of the block encoded matrices $\mW_{\lambda}$ and $\mI$ yields the block encoding of $\mW_{\lambda}+\frac{\Delta}{\gamma}\mI$, which has the runtime
    \begin{equation}
        \label{eq:runtime_block_encoding_4}
        \widetilde{O}\Big(\frac{D}{\lambda}\Big).
    \end{equation}

    To apply ${(\mW_{\lambda}+\frac{\Delta}{\gamma})\mI}^{-\frac{1}{2}}$ in Step~4,
    we use QSVT of the block encoding of $\mW_{\lambda}+\frac{\Delta}{\gamma}\mI$.
    In the same way as applying $\mA^{-1}$ explained above, given any state $\ket{\psi}$ and a circuit $\mU$ for a block encoding of $\mA$, using the technique of variable-time amplitude amplification, we can prepare a state proportional to $\mA^{-\frac{1}{2}}\ket{\psi}$ by querying $\mU$ and $\mU^\dag$ in total $\widetilde{O}(\kappa)$ times~\citep{10.1145/3313276.3316366}.
    In the same way as the analysis of Step~3 with replacing $\lambda$ with $\frac{\Delta}{\gamma}$,
    the condition number of $\mW_{\lambda}+\frac{\Delta}{\gamma}\mI$ is $O(\frac{1}{\nicefrac{\Delta}{\gamma}})$.
    As a whole, the runtime by the end of Step~4 is dominated by
    \begin{equation}
    \label{eq:runtime_step4}
        \widetilde{O}\Big(\frac{1}{\nicefrac{\Delta}{\gamma}}\Big)\times\widetilde{O}\Big(\frac{D}{\lambda}\Big).
    \end{equation}

    Consequently, the overall runtime of Algorithm~\ref{alg:sampling} is bounded by~\eqref{eq:runtime_step4} in Step~4, i.e.,
    \begin{equation}
        \widetilde{O}\left(\frac{D}{\lambda}\frac{1}{\nicefrac{\Delta}{\gamma}}\right)=\widetilde{O}\left(\frac{D}{\lambda\Delta}\times\gamma\right),
    \end{equation}
    which yields the conclusion.
\end{proof}

\section{\label{sec:E}Proof on Bounds for Finding Winning Ticket of Neural Networks}

\begin{algorithm}[t]
  \caption{\label{alg:winning}Quantum algorithm for finding a winning ticket for the original neural network $\mathcal{S}[w_\lambda^\ast]$.}
  \begin{algorithmic}[1]
      \REQUIRE{$\epsilon,\delta,\lambda>0$, assumptions in Sec.~\ref{sec:setting}, $M$ examples stored in data structure in Sec~\ref{sec:winning}, $\Delta$ and $N$ bounded by Theorem~\ref{thm:winning}.}
      \ENSURE{A subnetwork $\hat{f}$ of $\mathcal{S}[w_\lambda^\ast]$ achieving~\eqref{eq:S_hat_f_error} and characterized as~\eqref{eq:sparse_subnetwork} by the set of parameters $\hat{\sW}\subset\sZ_P^D\times\sZ_P$ and the weights $\boldsymbol{\hat{w}}^\ast$ for the nodes in the hidden layer.}
      \STATE{Initialize $\hat{\sW}=\emptyset$.}
      \FOR{$N$ times}
      \STATE{Sample $(\va,b)\in\sZ_P^D\times\sZ_P$ from $p_{\lambda,\Delta}^\ast(\va,b)$ in~\eqref{eq:optimized_distribution} by the quantum algorithm in Theorem~\ref{thm:quantum_sampling}.}
      \STATE{Add the sampled parameter $(\va,b)$ to $\hat{\sW}$}.
      \ENDFOR%
      \STATE{Train the weight $w(\va,b)$ of the sampled subnetwork $\sum_{(\va,b)\in\hat{\sW}}w(\va,b)g((\va^\top\vx-b)\bmod P)$ by convex optimization using the given $M$ examples, to obtain $\hat{\vw}^\ast$ as a solution of minimizing $\tilde{J}(w)$ in~\eqref{eq:minimization}.}
    \COMMENT{This convex optimization can be performed by classical algorithms such as stochastic gradient descent (SGD)~\citep{pmlr-v99-harvey19a}.}
    \STATE{\textbf{Return } $\hat{\sW}$ and $\hat{\vw}^\ast$.}
  \end{algorithmic}
\end{algorithm}

We provide the proof of Theorem~\ref{thm:winning} in Sec.~\ref{sec:finding}.
Our algorithm for Theorem~\ref{thm:winning} is shown in Algorithm~\ref{alg:winning}.

\begin{proof}[Proof of Theorem~\ref{thm:winning}]
With writing
\begin{equation}
\label{eq:N_Delta}
    N_\epsilon\coloneqq{\Big(\frac{\alpha}{\beta\sqrt{\epsilon}}\Big)}^{\frac{1}{\beta}},
\end{equation}
we set $\Delta$ and $N$ as
\begin{align}
    \label{eq:delta_epsilon}
    \Delta&={(\alpha {(N_\epsilon+1)}^{-(1+\beta)})}^2=\Omega(\epsilon^{1+\frac{1}{\beta}}),\\
\label{eq:N_optimal}
    N&=\left\lceil2\left(\lceil N_\epsilon\rceil+\frac{1}{1+2\beta}\frac{{(N_\epsilon+1)}^{2+2\beta}}{N_\epsilon^{1+2\beta}}\right)\ln\Big(\frac{\lceil N_\epsilon\rceil}{\delta}\Big)\right\rceil=O\Big(N_\epsilon\log\Big(\frac{N_\epsilon}{\delta}\Big)\Big)=O\left(\frac{1}{\epsilon^{\frac{1}{2\beta}}}\log\Big(\frac{1}{\epsilon\delta}\Big)\right).
\end{align}
Note that we can set $\Delta$ and $N$ flexibly up to changing $\alpha$ and $\beta$ in the constant factors of~\eqref{eq:delta_epsilon} and~\eqref{eq:N_optimal}, as long as the function $f$ to be learned remains in the set of $(\alpha,\beta)$-class functions.
The parameter $N_\epsilon$ is chosen in such a way that we have
\begin{equation}
    \sum_{j=\lceil N_\epsilon\rceil+1}^{\infty}\alpha j^{-(1+\beta)}
    \leqq\int_{N_\epsilon}^{\infty}\alpha x^{-(1+\beta)}dx=\frac{\alpha}{\beta N_\epsilon^\beta}=\sqrt{\epsilon}.
\end{equation}
For the analysis, let $\sW_{\lambda,\Delta}$ denote a set of parameters $(\va,b)$ for high-weight nodes in the hidden layer of $\mathcal{S}[w_\lambda^\ast]$, given by
\begin{align}
\label{eq:winning}
    &\sW_{\lambda,\Delta}\coloneqq\{(\va_j,b_j)\in \sZ_P^D\times\sZ_P:|P^{-\frac{D}{2}}w_\lambda^\ast(\va_j,b_j)|^2\geqq\Delta, j\in\{1,\ldots,P^{D+1}\} \}.
\end{align}
Then, due to the assumption on the $(\alpha,\beta)$-class functions
\begin{equation}
    |P^{-\frac{D}{2}}w_\lambda^\ast(\va_j,b_j)|\leqq\alpha j^{-(1+\beta)},
\end{equation}
we have by construction
\begin{equation}
    \label{eq:W_bound}
    |\sW_{\lambda,\Delta}|\leqq \lceil N_\epsilon\rceil.
\end{equation}

For any $(\va,b)\in\sW_{\lambda,\Delta}$, we have
\begin{align}
    &\Pr\{(\va,b)\not\in\hat{\sW}\}\\
    &={(1-p_{\lambda,\Delta}^\ast(\va,b))}^N\\
    &={\Big(1-\frac{1}{Z}\frac{|P^{-\frac{D}{2}}w_{\lambda}^\ast(\va,b)|^2}{|P^{-\frac{D}{2}}w_{\lambda}^\ast(\va,b)|^2+\Delta}\Big)}^N\\
    &\leqq{\Big(1-\frac{1}{|\sW_{\lambda,\Delta}|+\frac{1}{1+2\beta}\frac{{(N_\epsilon+1)}^{2+2\beta}}{N_\epsilon^{1+2\beta}}}\frac{\Delta}{\Delta+\Delta}\Big)}^N\\
    &\leqq\exp\left(-\frac{N}{2\left(|\sW_{\lambda,\Delta}|+\frac{1}{1+2\beta}\frac{{(N_\epsilon+1)}^{2+2\beta}}{N_\epsilon^{1+2\beta}}\right)}\right),
\end{align}
where the first inequality follows from
\begin{align}
    Z&=\sum_{(\va,b)\in\sW_{\lambda,\Delta}}\frac{|P^{-\frac{D}{2}}w_\lambda^\ast(\va,b)|^2}{|P^{-\frac{D}{2}}w_\lambda^\ast(\va,b)|^2+\Delta}+ \sum_{(\va,b)\in\sZ_P^D\times\sZ_P\setminus\sW_{\lambda,\Delta}}\frac{|P^{-\frac{D}{2}}w_\lambda^\ast(\va,b)|^2}{|P^{-\frac{D}{2}}w_\lambda^\ast(\va,b)|^2+\Delta}\\
     &\leqq|\sW_{\lambda,\Delta}|+\frac{1}{\Delta}\sum_{j=\lceil N_\epsilon\rceil+1}^{\infty}{|\alpha j^{-(1+\beta)}|}^2\\
     &=|\sW_{\lambda,\Delta}|+\frac{1}{\Delta}\sum_{j=\lceil N_\epsilon\rceil+1}^{\infty}{\alpha^2 j^{-2(1+\beta)}}\\
     &\leqq|\sW_{\lambda,\Delta}|+\frac{1}{\Delta}\int_{N_\epsilon}^\infty\frac{\alpha^2}{x^{-2(1+\beta)}}dx\\
     &=|\sW_{\lambda,\Delta}|+\frac{1}{\Delta}\frac{\alpha^2}{(1+2\beta)N_\epsilon^{1+2\beta}}\\
     &=|\sW_{\lambda,\Delta}|+\frac{1}{1+2\beta}\frac{{(N_\epsilon+1)}^{2+2\beta}}{N_\epsilon^{1+2\beta}},
\end{align}
and the second inequality follows from $1-x\leqq\mathrm{e}^{-x}$ for $x\in\sR$.
Thus, due to the union bound,
it follows from~\eqref{eq:N_optimal} and~\eqref{eq:W_bound} that
\begin{align}
    &\Pr\{\sW_{\lambda,\Delta}\not\subseteqq\hat{\sW}\}\\
    &\leqq|\sW_{\lambda,\Delta}|\exp\left(-\frac{N}{2\left(|\sW_{\lambda,\Delta}|+\frac{1}{1+2\beta}\frac{{(N_\epsilon+1)}^{2+2\beta}}{N_\epsilon^{1+2\beta}}\right)}\right)\\
    &\leqq\delta\frac{|\sW_{\lambda,\Delta}|}{\lceil N_\epsilon\rceil}\exp\left(-\frac{2\left(\lceil N_\epsilon\rceil+\frac{1}{1+2\beta}\frac{{(N_\epsilon+1)}^{2+2\beta}}{N_\epsilon^{1+2\beta}}\right)}{2\left(|\sW_{\lambda,\Delta}|+\frac{1}{1+2\beta}\frac{{(N_\epsilon+1)}^{2+2\beta}}{N_\epsilon^{1+2\beta}}\right)}\right)\\
    &\leqq\delta.
\end{align}

Therefore, with the choice of $N$ in~\eqref{eq:N_optimal}, we have $\sW_{\lambda,\Delta}\subseteqq\hat{\sW}$ with high probability greater than $1-\delta$.
Then for any $\hat{\sW}$ satisfying $\sW_{\lambda,\Delta}\subseteqq\hat{\sW}$, 
the subnetwork with nodes in the hidden layer parameterized by $(\va,b)\in\hat{\sW}$ approximates the original network $\mathcal{S}[w_\lambda^\ast]$ as
\begin{align}
    &\inf_w\left\{\sum_{\vx}\hat{p}_\mathrm{data}(\vx)\left|\mathcal{S}[w_\lambda^\ast](\vx)-\sum_{(\va,b)\in\hat{\sW}}P^{-\frac{D}{2}}w(\va,b)g((\va^{\top}\vx-b)\bmod P)\right|^2\right\}\\
    &\leqq\sum_{\vx}\hat{p}_\mathrm{data}(\vx)\left|\mathcal{S}[w_\lambda^\ast](\vx)-\sum_{(\va,b)\in\sW_{\lambda,\Delta}}P^{-\frac{D}{2}}w_\lambda^\ast(\va,b)g((\va^{\top}\vx-b)\bmod P)\right|^2\\
    &=\sum_{\vx}\hat{p}_\mathrm{data}(\vx)\left|\sum_{(\va,b)\in\sZ_P^D\times\sZ_P}P^{-\frac{D}{2}}w_\lambda^\ast(\va,b)g((\va^{\top}\vx-b)\bmod P)-\sum_{(\va,b)\in\sW_{\lambda,\Delta}}P^{-\frac{D}{2}}w_\lambda^\ast(\va,b)g((\va^{\top}\vx-b)\bmod P)\right|^2\\
    &=\sum_{\vx}\hat{p}_\mathrm{data}(\vx)\left|\sum_{(\va,b)\in\sZ_P^D\times\sZ_P\setminus\sW_{\lambda,\Delta}}P^{-\frac{D}{2}}w_\lambda^\ast(\va,b)g((\va^{\top}\vx-b)\bmod P)\right|^2\\
    &\leqq\sum_{\vx}\hat{p}_\mathrm{data}(\vx){\left(\sum_{(\va,b)\in\sZ_P^D\times\sZ_P\setminus\sW_{\lambda,\Delta}}\left|P^{-\frac{D}{2}}w_\lambda^\ast(\va,b)\right|\left|g((\va^{\top}\vx-b)\bmod P)\right|\right)}^2\\
    &\leqq\sum_{\vx}\hat{p}_\mathrm{data}(\vx){\left(\sum_{(\va,b)\in\sZ_P^D\times\sZ_P\setminus\sW_{\lambda,\Delta}}\left|P^{-\frac{D}{2}}w_\lambda^\ast(\va,b)\right|\right)}^2\\
    &={\left(\sum_{(\va,b)\in\sZ_P^D\times\sZ_P\setminus\sW_{\lambda,\Delta}}\left|P^{-\frac{D}{2}}w_\lambda^\ast(\va,b)\right|\right)}^2\\
    &\leqq{\left(\sum_{j=\lceil N_\epsilon\rceil+1}^{\infty}\alpha j^{-(1+\beta)}\right)}^2\\
    &\leqq\epsilon,
\end{align}
which yield the conclusion.
\end{proof}

\section{\label{sec:F}Detail of Numerical Experiment on Advantage of Using Quantum Ridgelet Transform}

\begin{figure}
    \centering
    \includegraphics[width=3.4in]{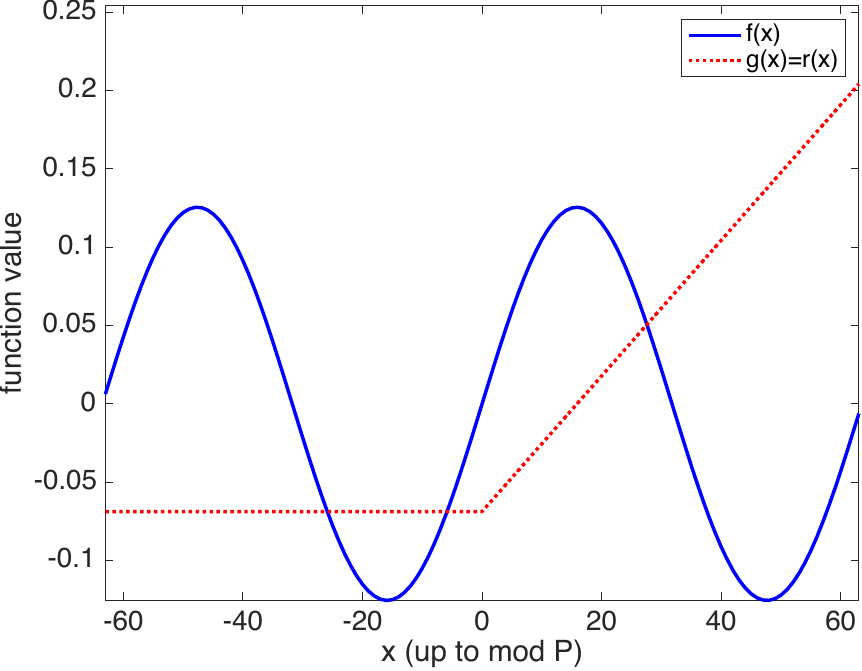}
    \caption{\label{fig:f_g_r}The function $f$ to be learned and the activation function $g$ chosen as ReLU, where the ridgelet function $r$ is also chosen as $r=g$ in our numerical experiment.}
\end{figure}

We describe the detail of the parameters chosen for the numerical experiment in Sec.~\ref{sec:advantage}.

\begin{figure}
    \centering
    \begin{minipage}[b]{0.45\linewidth}
    \centering
    \includegraphics[width=\linewidth]{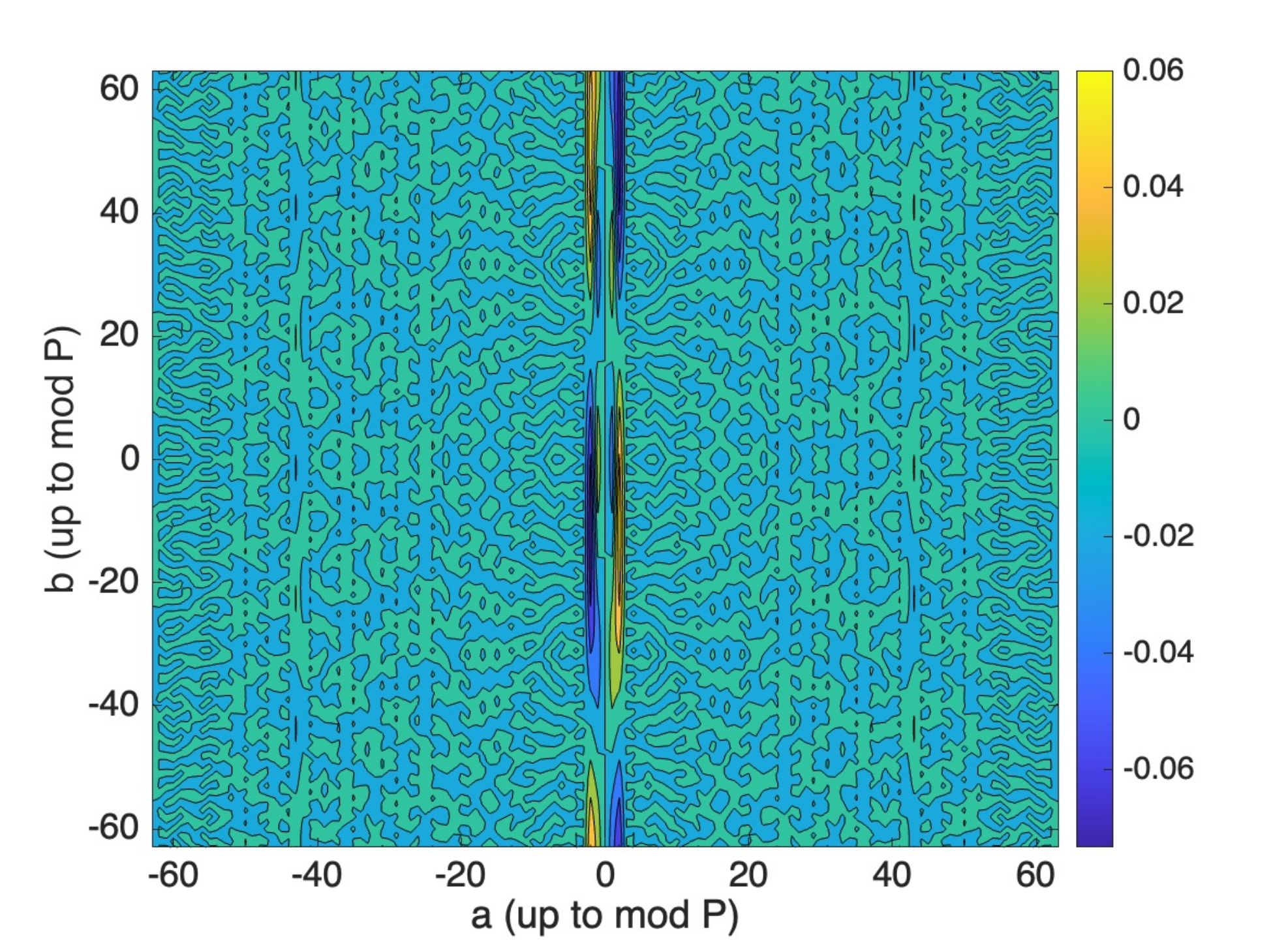}
  \end{minipage}
  \begin{minipage}[b]{0.45\linewidth}
    \centering
    \includegraphics[width=\linewidth]{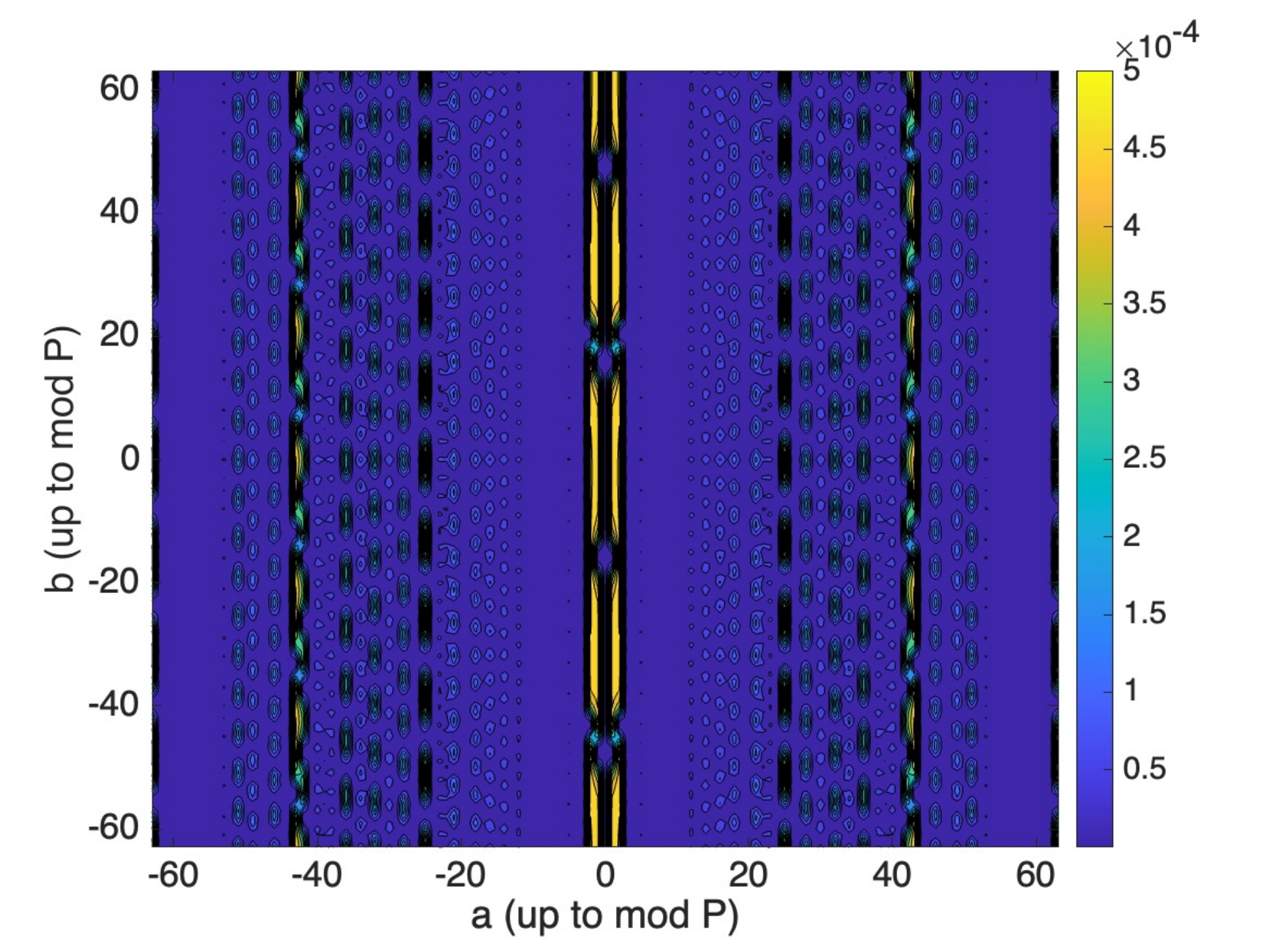}
  \end{minipage}
    \caption{\label{fig:R}The discrete ridgelet transform $\mathcal{R}[f](a,b)$ of $f$ in Fig.~\ref{fig:f_g_r} on the left, and the optimized probability distribution $p_{\lambda,\Delta}^\ast(a,b)$ on the right.}
\end{figure}

In our numerical experiment, we fixed
\begin{align}
    P=127,\\
    D=1.
\end{align}
In the current case of $D=1$, we may write $\vx$ and $\va$ as $x$ and $a$, respectively.
Up to normalization, we chose
\begin{align}
    \label{eq:f_definition}
    f(x)&\propto\sin\left(\frac{4\pi x}{P}\right),\\
    g(x)=r(x)&\propto\begin{cases}
        x\bmod P&0\leqq (x\bmod P)\leqq\frac{P-1}{2},\\
        0&\frac{P+1}{2}\leqq (x\bmod P)\leqq P-1,
    \end{cases}
\end{align}
which are shown in Fig.~\ref{fig:f_g_r} with normalization.
The empirical distribution was chosen as the uniform distribution
\begin{equation}
    \hat{p}_\mathrm{data}(x)=\frac{1}{P}.
\end{equation}
The parameters in our setting were chosen as
\begin{align}
        \epsilon&=5\times10^{-2},\\
        \lambda&=1\times10^{-4},\\
        \alpha&=4\times10^{21},\\
        \beta&=5,
\end{align}
which leads to
\begin{equation}
        \gamma\approx 9.8\times 10^{-1}.\\
\end{equation}
According to~\eqref{eq:N_Delta},~\eqref{eq:delta_epsilon}, and~\eqref{eq:N_optimal}, it suffices to set
\begin{align}
        N_\epsilon&\approx 2.0\times 10^{4},\\
        \Delta&\approx 5.5\times 10^{-5},\\
        N&\approx 5.9\times 10^{5}.
\end{align}

With the above choice of parameters, the discrete ridgelet transform $\mathcal{R}[f](a,b)$ of $f$ and the optimized probability distribution $p_{\lambda,\Delta}^\ast(a,b)$ are calculated and illustrated in Fig.~\ref{fig:R}.
In our numerical experiment, the sampling from $p_{\lambda,\Delta}^\ast$ is classically simulated by rejection sampling using the calculated value of $p_{\lambda,\Delta}^\ast(a,b)$ for each $a$ and $b$.
The uniform distribution $p_\mathrm{uniform}(a,b)=\frac{1}{P^{D+1}}$ is used for sampling the random features.
After performing the sampling $N$ times for $N=4,8,\ldots,120$, we optimize $\hat{w}^\ast(a,b)$ in~\eqref{eq:sparse_subnetwork} by minimizing the empirical risk $\sum_{x}\hat{p}_\mathrm{data}(x)|f(x)-\sum_{(a,b)\in\hat{\sW}}P^{-\frac{D}{2}}\hat{w}^\ast(a,b)g((ax-b)\bmod P)|^2$.
This convex optimization was solved by MOSEK~\citep{mosek} and YALMIP~\citep{Lofberg2004}.
The result of this numerical experiment is presented in Fig.~\ref{fig:advantage} of Sec.~\ref{sec:advantage} in the main text.

To support our results further,
we also performed numerical simulation in the case of choosing $g$ ($=r$) as a sigmoid function rather than ReLU used in Fig.~\ref{fig:advantage};
in particular, in this additional numerical simulation, with the same $f$ as~\eqref{eq:f_definition}, we chose, up to normalization,
\begin{align}
    g(x)=r(x)&\propto\tanh\left(\frac{10x}{P}\right),
\end{align}
which are shown at the top of Fig.~\ref{fig:f_g_r_tanh} with normalization.
Setting the other parameters the same as Fig.~\ref{fig:advantage}, we also present the result of this additional numerical experiment at the bottom of Fig.~\ref{fig:f_g_r_tanh}.

These numerical results show that the actual performance of empirical risk minimization in our numerical experiments can be much better than the theoretically guaranteed upper bound on $N$ in Theorem~\ref{thm:winning}, implying even more actual advantage of our algorithm than the theoretically guaranteed bound in our setting.

\begin{figure}
    \centering
    \includegraphics[width=3.4in]{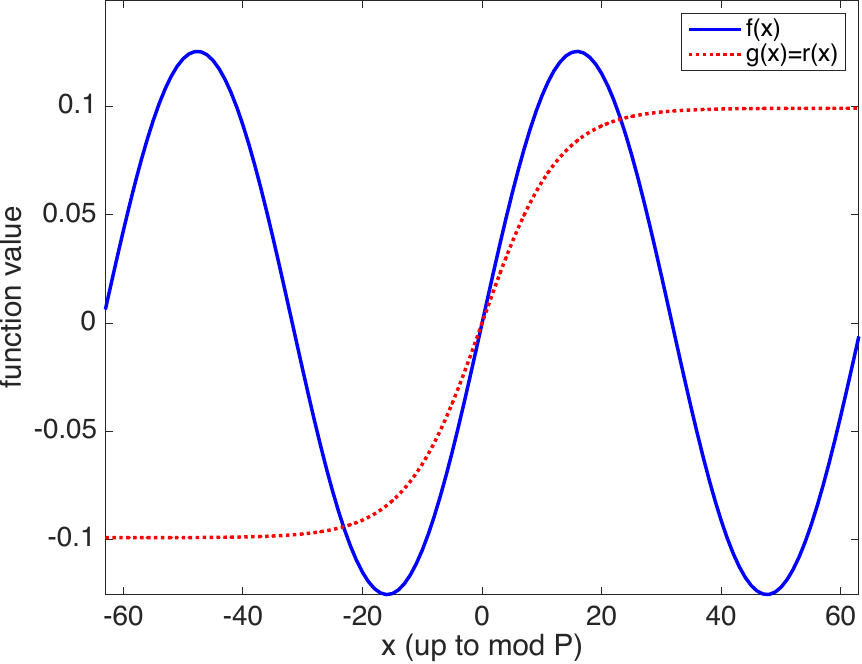}\\
    \includegraphics[width=3.4in]{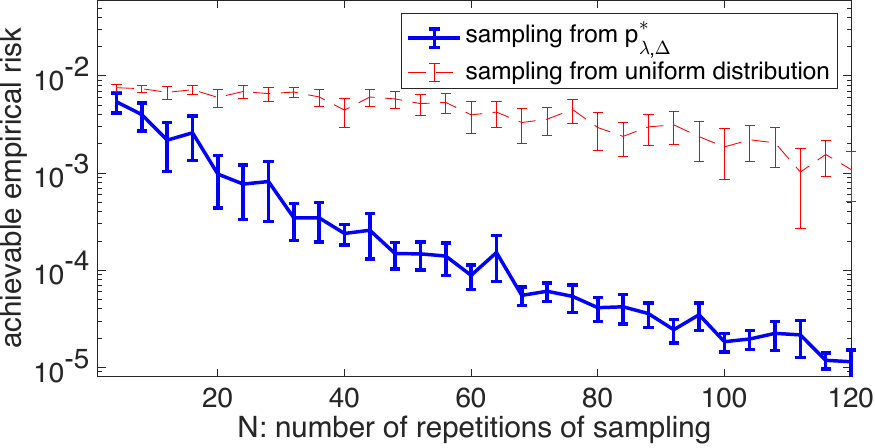}
\caption{\label{fig:f_g_r_tanh}The numerical experiment performed in addition to the one presented in Fig.~\ref{fig:advantage} of the main text. In this additional numerical experiment, as shown at the top of the figure, the function $f$ to be learned is the same as the one in Fig.~\ref{fig:f_g_r}, whereas the activation function $g$ and the ridgelet function $r$ are chosen as a sigmoid function (tanh) as shown in the figure, rather than ReLU in Fig.~\ref{fig:f_g_r}. Setting the other parameters the same as the numerical experiment in Fig.~\ref{fig:advantage}, at the bottom of the figure, we plot the empirical risks achieved by sampling from the optimized distribution via our algorithm (blue thick line) and that from the uniform distribution via random features (red dashed line), which support our results further in the same way as Fig.~\ref{fig:advantage}.}
\end{figure}

\end{document}